\documentclass[article]{IEEEtran}
\usepackage[top=1in, bottom=1in, left=0.75in, right=0.75in]{geometry}

\usepackage{amsmath, amssymb, epsfig, graphicx, bm}
\usepackage{cases}
\usepackage{float,subfigure}
\usepackage{color}
\usepackage{amsthm}
\usepackage{amsfonts}
\usepackage{cite, url}
\usepackage[all,cmtip]{xy}
\usepackage{color,hyperref}
\usepackage{algorithmic,algorithm}

\definecolor{darkblue}{rgb}{0,0,1}
\hypersetup{colorlinks,breaklinks,
linkcolor=darkblue,urlcolor=darkblue,anchorcolor=darkblue,citecolor=darkblue}

\usepackage[normalem]{ulem} 

\newtheorem{theorem}{Theorem}

\newtheorem{lemma}{Lemma}

\newcommand{\R}{{\mathbb{R}}}%
\newcommand{\bu}{{\mathbf{u}}}%
\newcommand{\bv}{{\mathbf{v}}}%
\newcommand{\s}{{\mathbf{s}}}%
\newcommand{\h}{\mathbf{h}}%
\newcommand{\bPhi}{\mathbf{\Phi}}%
\newcommand{\z}{\mathbf{z}}%
\newcommand{\barz}{\bar{\mathbf{z}}}%
\newcommand{\tildez}{\tilde{\mathbf{z}}}%
\newcommand{\La}{\mathcal{L}}%
\newcommand{\N}{\mathcal{N}}%
\newcommand{\bpsi}{\bm{\psi}}%
\newcommand{\bdelta}{\bm{\delta}}%
\newcommand{\M}{\mathbf{M}}
\newcommand{\be}{\mathbf{e}}
\newcommand{\Rhj}{\mathbf{R}_{h_j}}

\newcommand{\beps}{\bm{\epsilon}}

\newcommand{\bC}{\mathbf{C}}
\newcommand{\bH}{\mathbf{H}}
\newcommand{\bL}{\mathbf{L}}

\begin{document}

\title{Decentralized RLS with Data-Adaptive Censoring for Regressions over Large-Scale Networks}
\author{\authorblockN{Zifeng Wang, Zheng Yu, Qing Ling, Dimitris Berberidis, and Georgios B. Giannakis}
\thanks{Zifeng Wang and Zheng Yu are with the Special Class for the Gifted Young, University of Science and Technology of China.
Emails: wzfeng@mail.ustc.edu.cn and yz123@mail.ustc.edu.cn. Qing
Ling (corresponding author) is with the School of Data and
Computer Science, Sun Yat-Sen University. Email:
lingqing556@mail.sysu.edu.cn. Dimitris Berberidis and Georgios B.
Giannakis are with the Dept. of ECE, University of Minnesota.
Emails: bermp001@umn.edu and georgios@umn.edu. This work was
supported by NSF China grant 61573331, NSF Anhui grant
1608085QF130, NSF grants 1442686, 1500713, and 1711471. Part of
this paper has appeared at the 42nd \emph{Intl. Conf. on
Acoustics, Speech, and Signal Processing}, New Orleans, USA, March
5--9, 2017 \cite{CD-RLS-ICASSP}.}}


\maketitle

\begin{abstract}

The deluge of networked data motivates the development of
algorithms for computation- and communication-efficient
information processing. In this context, three data-adaptive
censoring strategies are introduced to considerably reduce the
computation and communication overhead of decentralized recursive
least-squares (D-RLS) solvers. The first relies on alternating
minimization and the stochastic Newton iteration to minimize a
network-wide cost, which discards observations with small
innovations. In the resultant algorithm, each node performs local
data-adaptive censoring to reduce computations, while exchanging
its local estimate with neighbors so as to consent on a
network-wide solution. The communication cost is further reduced
by the second strategy, which prevents a node from transmitting
its local estimate to neighbors when the innovation it induces to
incoming data is minimal. In the third strategy, not only
transmitting, but also receiving estimates from neighbors is
prohibited when data-adaptive censoring is in effect. For all
strategies, a simple criterion is provided for selecting the
threshold of innovation to reach a prescribed average data
reduction. The novel censoring-based (C)D-RLS algorithms are
proved convergent to the optimal argument in the mean-root
deviation sense. Numerical experiments validate the effectiveness
of the proposed algorithms in reducing computation and
communication overhead.

\end{abstract}


\section{Introduction} \label{sec:intro}

In our big data era, various networks generate massive amounts of
streaming data. Examples include wireless sensor networks, where a
large number of inexpensive sensors cooperate to monitor, e.g. the
environment \cite{Predd2006, Rabbat2004-ipsn}, or data centers,
where a group of servers collaboratively handles dynamic user
requests \cite{Sharkh2013}. Since a single node has limited
computational resources, decentralized information processing is
preferable as the network size scales up \cite{Cevher2014,
Giannakis2016}. In this paper, we focus on a decentralized linear
regression setup, and develop computation- and
communication-efficient decentralized recursive least-squares
(D-RLS) algorithms.

The main tool we adopt to reduce computation and communication
costs is data-adaptive censoring, which leverages the redundancy
present especially in big data. Upon receiving an observation, nodes determine
whether it is informative or not. Less informative observations
are discarded, while messages among neighboring
nodes are exchanged only when necessary. We propose three censoring-based
(C)D-RLS algorithms that can achieve estimation accuracy
comparable to D-RLS without censoring, while
significantly reducing the computation and communication overhead.

\subsection{Related works}
\label{sec:intro-survey}

The merits of RLS algorithms in solving centralized linear
regression problems are well recognized
\cite{Kushner1997-RLS,Slavakis2014-RLS}. When streaming observations
that depend linearly on a set of unknown parameters become available,
RLS yields the least-squares parameter estimates online.
RLS reduces the computational burden of finding a batch
estimate per iteration, and can even allow for tracking
time-varying parameters. The computational cost can be further
reduced by data-adaptive censoring \cite{Berberidis2016-Censoring},
where less informative data are discarded. On the other hand,
decentralized versions of RLS without censoring have been advocated to
solve linear regression tasks over networks \cite{Mateos2012-DRLS}.
In D-RLS, a node updates its estimate that is common to the
entire network by fusing its
local observations with the local estimates of its neighbors. As
time evolves, all local estimates consent on the centralized RLS
solution. This paper builds on both
\cite{Berberidis2016-Censoring} and \cite{Mateos2012-DRLS} by
developing censoring-based decentralized RLS algorithms, thus
catering to efficient online linear regression over large-scale
networks.

Different from our in-network setting where operation is fully
decentralized and nodes are only able to communicate with their
neighbors, most of the existing distributed censoring algorithms
apply to star topology networks that rely on a fusion center
\cite{Appadwedula2008-Censoring,Jiang2005-censoring,Jiang2005-onoff,Msechu2012-cratio,Rago1996-censoring}.
Their basic idea is that each node transmits data to the fusion
center for further processing only when its local likelihood ratio
exceeds a threshold \cite{Rago1996-censoring}; see also
\cite{Jiang2005-censoring} where communication constraints are
also taken into account. Information fusion over fading channels
is considered in \cite{Jiang2005-onoff}. Practical issues such as
joint dependence of sensor decision rules, randomization of
decision strategies as well as partially known distributions are
reported in \cite{Appadwedula2008-Censoring}, while
\cite{Msechu2012-cratio} also explores quantization jointly with
censoring.

Other than the star topology studied in the aforementioned works,
\cite{Patwari2003-censoring} investigates censoring for a tree
structure. If a node's local likelihood ratio exceeds a threshold,
its local data is sent to its parent node for fusion. A fully
decentralized setting is considered in
\cite{Arroyo2013-censoring}, where each node determines whether to
transmit its local estimate to its neighbors by comparing the
local estimate with the weighted average of its neighbors.
Nevertheless, \cite{Arroyo2013-censoring} aims at mitigating only
the communication cost, while the present work also considers
reduction of the computational cost across the network.
Furthermore, the censoring-based decentralized linear regression
algorithm in \cite{Liu2015-censored-regression} deals with optimal
full-complexity estimation when observations are partially known
or corrupted. This is different from our context, where censoring
is deliberately introduced to reduce computational and
communication costs for decentralized linear regression.

\subsection{Our contributions and organization}
\label{sec:intro-cont}

The present paper introduces three data-adaptive online censoring
strategies for decentralized linear regression. The resultant
CD-RLS algorithms incur low computational and communication costs,
and are thus attractive for large-scale network applications
requiring decentralized solvers of linear regressions. Unlike most
related works that specifically target wireless sensor networks
(WSNs), the proposed algorithms may be used in a broader context
of decentralized linear regression using multiple computing
platforms. Of particular interest are cases where a regression
dataset is not available at a single machine, but it is
distributed over a network of computing agents that are interested
in accurately estimating the regression coefficients in an
efficient manner.

In Section \ref{sec:algo}, we formulate the decentralized online
linear regression problem (Section \ref{sec:algo-problem}), and
recast the D-RLS in \cite{Mateos2012-DRLS} into a new form (Section
\ref{sec:d-rls}) that prompts the development of three censoring
strategies (Section \ref{sec:c-d-rls}). Section
\ref{sec:theoretical} develops the first censoring strategy
(Section \ref{sec:derivation}), analyzes all three censoring
strategies (Section \ref{sec:convergence}), and discusses how to
set the censoring thresholds (Section \ref{sec:threshold}).
Numerical experiments in Section \ref{sec:numerical} demonstrate
the effectiveness of the novel CD-RLS algorithms.

\textit{Notation.} Lower (upper) case boldface letters denote
column vectors (matrices). $(\cdot)^T$, $||\cdot||$, $||\cdot||_2$
and $E[\cdot]$ stand for transpose, 2-norm, induced matrix 2-norm
and expectation, respectively. Symbols $\mathrm{tr}(\mathbf{X})$,
$\lambda_\text{min}(\mathbf{X})$ and
$\lambda_\text{max}(\mathbf{X})$ are used for the trace, minimum
eigenvalue and maximal eigenvalue of matrix $\mathbf{X}$,
respectively. Kronecker product is denoted by $\otimes$ and the
uniform distribution over $[a,b]$ by $\mathcal{U}(a,b)$, and the
Gaussian probability  distribution function (pdf) with mean $\mu$
and variance $\sigma^2$ by $\N(\mu,\sigma^2)$. The standardized
Gaussian pdf is $\phi(t) = (1/\sqrt{2\pi})\text{exp}(-t^2/2)$, and
its the associated complementary cumulative distribution function
is represented by $Q(z):=\int_{z}^{+\infty}\phi(t)dt$.

\section{Context and Algorithms}
\label{sec:algo}

This section outlines the online linear regression setup over
networks, and takes a fresh look at the D-RLS algorithm. Three
strategies are then developed using data-adaptive censoring to
reduce the computational and communication costs of D-RLS.

\subsection{Problem statement} \label{sec:algo-problem}

Consider a bidirectionally connected network with $J$ nodes,
described by a graph $\mathcal{G}:=\{\mathcal{V},\mathcal{E}\}$,
where $\mathcal{V}$ is the set of nodes with cardinality
$|\mathcal{V}|=J$, and $\mathcal{E}$ denotes the set of edges.
Each node $j$ only communicates with its one-hop neighbors,
collected in the set $\N_j \subset \mathcal{V}$. The decentralized
network is deployed to estimate a real vector $\s_0\in\R^{p}$. Per
time slot $t = 1,2,\ldots$, node $j$ receives a real scalar
observation $x_j(t)$ involving the wanted $\s_0$ with a regression
row $\h_j^T(t)$, so that $x_j(t)=\h_j^T(t)\s_0+\epsilon_j(t)$,
with $\epsilon_j(t)\sim \mathcal{N}(0,\sigma_j^2)$.

Our goal is to devise efficient decentralized online algorithms to
solve the following exponentially-weighted least-squares (EWLS)
problem
\begin{align}\label{eq:ewls}
\hat{\s}_{ewls}(t):=\arg\min\limits_{\s}
 \frac{1}{2} ~\sum_{r=1}^{t}\sum_{j=1}^{J}\lambda^{t-r} [x_j(r)-\h_j^T(r)\s]^2
\end{align}
where $\hat{\s}_{ewls}(t)$ is the EWLS estimate at slot
$t$, and $\lambda\in(0,1]$ is a forgetting factor that de-emphasizes the importance of past measurements, and thus enables tracking of a
non-stationary process. When $\lambda=1$, \eqref{eq:ewls} boils
down to a standard decentralized online least-squares estimate.

\subsection{D-RLS revisited} \label{sec:d-rls}

The D-RLS algorithm of \cite{Mateos2012-DRLS} solves
\eqref{eq:ewls} as follows. Per time slot $t$, node $j$ receives
$x_j(t)$ and $\h_j^T(t)$ and uses them to update the per-node inverse $p\times p$ covariance matrix as
\begin{align}
\bPhi_j^{-1}(t) &= \lambda^{-1}\bPhi_j^{-1}(t-1) \nonumber\\
&
-\frac{\lambda^{-1}\bPhi_j^{-1}(t-1)\h_j(t)\h_j^T(t)\bPhi_j^{-1}(t-1)}{\lambda+\h_j^T(t)\bPhi_j^{-1}(t-1)\h_j(t)}\label{eq:d-rls-phi}
\end{align}
along with the per-node $p\times 1$ cross-covariance vector as
\begin{align}
\bm{\psi}_j(t) &= \lambda\bm{\psi}_j(t-1)+\h_j(t)x_j(t).
\label{eq:d-rls-psi}
\end{align}
Using $\bPhi_j^{-1}(t)$ and $\bm{\psi}_j(t)$, node $j$ then updates its
local parameter estimate using
\begin{align}
\s_j(t) = \bPhi_j^{-1}(t) \Big[ \bm{\psi}_j(t)
-\frac{1}{2}\sum_{j'\in\N_j} \left(
\bv_j^{j'}(t-1)-\bv_{j'}^j(t-1) \right) \Big] \label{eq:d-rls-s}
\end{align}
where $\bv_j^{j'}(t-1)$ denotes the Lagrange multiplier of node
$j$ corresponding to its neighbor $j'$ at slot $t-1$, that
captures the accumulated differences of neighboring estimates,
recursively obtained as ($\rho >0$ is a step-size)
\begin{align}
\bv_j^{j'}(t-1) =
\bv_j^{j'}(t-2)+\rho\big[\s_j(t-1)-\s_{j'}(t-1)\big].
\label{eq:d-rls-v}
\end{align}

Next, we develop an equivalent novel form of D-RLS recursions
\eqref{eq:d-rls-phi}--\eqref{eq:d-rls-v} that is convenient for
our incorporation of data-adaptive censoring. Detailed derivation
of the equivalence can be found in Appendix \ref{sec:eq-DRLS}. The
inverse covariance matrix is updated as in \eqref{eq:d-rls-phi}. However, the
update of $\s_j(t)$ in \eqref{eq:d-rls-s} is replaced by
\begin{align}
\s_j(t)& =  \s_j(t-1)+\bPhi_j^{-1}(t)\h_j(t)\big[x_j(t)-\h_j^T(t)\s_j(t-1)\big]\nonumber\\
& -\rho\bPhi_j^{-1}(t)\bdelta_j(t-1)\label{eq:new-d-rls-s}
\end{align}
where $\bdelta_j(t)$ stands for a Lagrange multiplier
conveying network-wide information that is updated as
\begin{align}
\bdelta_j(t) & =  \bdelta_j(t-1) + \sum_{j'\in\N_j}[\s_j(t)-\s_{j'}(t)] \nonumber \\
& - \lambda
\sum_{j'\in\N_j}[\s_j(t-1)-\s_{j'}(t-1)].\label{eq:new-d-rls-delta}
\end{align}
Observe that $\bdelta_j(t)$ stores the weighted sum of differences
between the local estimate of node $j$, and all estimates of its
neighbors. Interestingly, if the network is disconnected and the
nodes are isolated, then $\bdelta_j(t) = \mathbf{0}$ so long as
$\bdelta_j(0) = \mathbf{0}$, and the update of $\s_j(t)$ in
\eqref{eq:new-d-rls-s} basically boils down to the centralized RLS
one \cite{Kushner1997-RLS,Slavakis2014-RLS}. That is, the current
estimate is modified from its previous value using the prediction
error $x_j(t)-\h_j^T(t)\s_j(t-1)$, which is known as the incoming
data \textit{innovation}. If on the other hand the network is
connected, nodes can leverage estimates of their neighbors
(captured by $\bdelta_j(t)$), which provide new information from
the network other than its own observations $\{x_j(t) \}$. The
term $\rho\bPhi_j^{-1}(t)\bdelta_j(t-1)$ can be viewed as a
Laplacian smoothing regularizer, which encourages all nodes of the
graph to reach consensus on their estimates.

\noindent {\bf Remark 1}. In D-RLS, \eqref{eq:d-rls-phi} incurs
computational complexity $O(p^2)$, since calculating the products
$\bPhi_j^{-1}(t-1)\h_j(t)$ and $\bPhi_j^{-1}(t-1)\bpsi_j(t)$
requires $O(p^2)$ multiplications. Similarly,
\eqref{eq:new-d-rls-s} incurs \emph{computational} complexity
$O(p^2)$, that is dominated by the matrix-vector multiplications
$\bPhi_j^{-1}(t)\h_j(t)$ and $\bPhi_j^{-1}(t)\bdelta_j(t-1)$. The
cost of carrying out \eqref{eq:new-d-rls-delta} is relatively
minor. Regarding \emph{communication} cost per slot $t$, node $j$
needs to transmit its local estimate $\s_j(t)$ to its neighbors
and receive estimates $\s_{j'}(t)$ from all neighbors $j' \in
\mathcal{N}_j$. The computational burden of D-RLS recursions
\eqref{eq:d-rls-phi}--\eqref{eq:d-rls-v} is comparable to that of
\eqref{eq:d-rls-phi}, \eqref{eq:new-d-rls-s} and
\eqref{eq:new-d-rls-delta}, with the cost of \eqref{eq:d-rls-s}
being the same as what \eqref{eq:new-d-rls-s} requires. Meanwhile,
the original form requires neighboring nodes $j$ and $j'$ to
exchange $\bv_j(t)$ and $\bv_{j'}(t)$ in addition to $\s_j(t)$ and
$\s_{j'}(t)$, which doubles the communication cost relative to
\eqref{eq:new-d-rls-s} and \eqref{eq:new-d-rls-delta}.

\subsection{Censoring-based D-RLS strategies} \label{sec:c-d-rls}

The D-RLS algorithm has well documented merits for decentralized
online linear regression \cite{Mateos2012-DRLS}. However, its
computational and communication costs per iteration are fixed,
regardless of whether observations and/or the estimates from
neighboring nodes are informative or not. This fact motivates our
idea of permeating benefits of data-adaptive censoring to
\emph{decentralized} RLS, through three novel censoring-based
(C)D-RLS strategies. They are different from the RLS algorithms in
\cite{Berberidis2016-Censoring}, where the focus is on
\emph{centralized} online linear regression.

Our first censoring strategy (CD-RLS-1) can be intuitively
motivated as follows. If a given datum $(x_j(t),\h_j(t))$ is not
informative enough, we do not have to use it since its
contribution to the local estimate of node $j$, as well as to
those of all network nodes, is limited. With $\{\tau\sigma_j(t)\}$
specifying proper thresholds to be discussed later, this intuition
can be realized using a censoring indicator variable
\begin{align}\label{eq:censor-indicator}
c_j(t) &:=
\begin{cases}
0,& \text{ if $|x_j(t)-\h_j^T(t)\s_j(t-1)| \leq \tau\sigma_j(t)$}\\
1,& \text{ if $|x_j(t)-\h_j^T(t)\s_j(t-1)| > \tau\sigma_j(t)$}.
\end{cases}
\end{align}
If the absolute value of the innovation is less than
$\tau\sigma_j(t)$, then $(x_j(t),\h_j(t))$ is censored;
otherwise $(x_j(t),\h_j(t))$ is used. Section \ref{sec:threshold}
will provide rules for selecting the threshold $\tau$ along with
the local noise variance $\sigma_j^2(t)$, whose
computations are lightweight. If data censoring is in effect, we
simply throw away the current datum by letting $\h_j(t) =
\mathbf{0}$ in \eqref{eq:d-rls-phi}, to obtain
\begin{align}
\bPhi_j^{-1}(t) &=
\lambda^{-1}\bPhi_j^{-1}(t-1).\label{eq:c-d-rls-phi}
\end{align}
Likewise, letting $x_j(t) = 0$ and $\h_j(t) = \mathbf{0}$ in
\eqref{eq:new-d-rls-s}, yields
\begin{align}
\s_j(t) =
\s_j(t-1)-\rho\bPhi_j^{-1}(t)\bdelta_j(t-1).\label{eq:c-d-rls-s}
\end{align}

CD-RLS-1 is summarized in Algorithm 1.
If censoring is in effect, computation cost per node and per
slot is a fraction $2/7$ of the D-RLS in \eqref{eq:d-rls-s}
and \eqref{eq:new-d-rls-delta} without censoring. To recognize
why, observe that the
scalar-matrix multiplication  $\lambda^{-1}\bPhi_j^{-1}(t-1)$
in \eqref{eq:c-d-rls-phi} is not necessary as the update of
$\bPhi_j^{-1}(t)$ can be merged to wherever it is needed, e.g.,
in \eqref{eq:c-d-rls-s} and the next slot. In addition,
carrying out the $O(p^2)$ multiplications to obtain
$\bPhi_j^{-1}(t)\h_j(t)$ is no longer necessary, while the $O(p^2)$
multiplications required to obtain $\bPhi_j^{-1}(t)\bdelta_j(t-1)$
remain the same.

The first censoring strategy still requires nodes to communicate
with neighbors per time slot; hence, the communication cost
remains the same. Reducing this communication cost, motivates our
second censoring strategy (CD-RLS-2), where each node does not
perform extra computations relative to CD-RLS-1, but only receives
neighboring estimates if its current datum is censored. The
intuition behind this strategy is that if a datum is censored,
then very likely the current local estimate is sufficiently
accurate, and the node does not need to account for estimates from
its neighbors. Estimates from neighbors, are only stored for
future usage. Likewise, neighbors in $\mathcal{N}_j$ do not need
node $j$'s current estimate either, because they have already
received a very similar estimate. CD-RLS-2 is summarized in
Algorithm 2.

The third censoring strategy (CD-RLS-3) given by Algorithm 3 is
more aggressive than the second one. If a node has its datum
censored at a certain slot, then it neither transmits to nor
receives from its neighbors, and in that sense it remains
``isolated'' from the rest of the network in this slot.
Apparently, we should not allow any node to be forever
isolated. To this end, we can force each node to receive
the local estimate from any of its neighbors at least once
every $d_{\max}$ slots, which upper bounds the delay of
information exchange to $d_{\max}$. Interestingly, the
ensuing section will prove convergence of all three strategies
to the optimal argument in the mean-square deviation sense
under mild conditions.

\begin{algorithm}[t]
{\small
\caption{CD-RLS-1}
\begin{algorithmic}[1]
\STATE Initialize $\bdelta_j(0)$, $\{\s_j(0)\}_{j=1}^{J}$ and
$\{\bPhi^{-1}_j(0)\}_{j=1}^{J}$ \FOR{$t=1,2,\ldots$}
   \STATE All $j\in\mathcal{V}$:
   \IF{$|x_j(t)-\h_j^T(t)\s_j(t-1)|\leq \tau\sigma_j(t)$}
   \STATE update $\bPhi^{-1}_j(t)$ using \eqref{eq:c-d-rls-phi}
   \STATE update $\s_j(t)$ using \eqref{eq:c-d-rls-s}
   \ELSE
   \STATE update $\bPhi^{-1}_j(t)$ using \eqref{eq:d-rls-phi}
   \STATE update $\s_j(t)$ using \eqref{eq:new-d-rls-s}
   \ENDIF
   \STATE transmit $\s_j(t)$ to and receive $\s_{j'}(t)$ from all $j' \in \N_j$
   \STATE compute $\bdelta_j(t)$ using \eqref{eq:new-d-rls-delta}
\ENDFOR
\end{algorithmic}}
\end{algorithm}

\begin{algorithm}[t]
{\small
\caption{CD-RLS-2}
\begin{algorithmic}[1]
\STATE Initialize $\bdelta_j(0)$, $\{\s_j(0)\}_{j=1}^{J}$ and
$\{\bPhi^{-1}_j(0)\}_{j=1}^{J}$ \FOR{$t=1,2,\ldots$}
   \STATE All $j\in\mathcal{V}$:
   \IF{$|x_j(t)-\h_j^T(t)\s_j(t-1)|\leq \tau\sigma_j(t)$}
   \STATE receives $\s_{j'}(t)$ from all $j' \in \N_j$
   \ELSE
   \STATE set $\s_{j'}(t-1)$ as recently received
   ones from all $j' \in \N_j$
   \STATE update $\bPhi^{-1}_j(t)$ using \eqref{eq:d-rls-phi}
   \STATE update $\s_j(t)$ using \eqref{eq:new-d-rls-s}
   \STATE transmit $\s_j(t)$ to and receive $\s_{j'}(t)$ from all $j' \in \N_j$
   \STATE compute $\bdelta_j(t)$ using \eqref{eq:new-d-rls-delta}
   \ENDIF
\ENDFOR
\end{algorithmic}}
\end{algorithm}

\begin{algorithm}[t]
{\small
\caption{CD-RLS-3}
\begin{algorithmic}[1]
\STATE Initialize $\bdelta_j(0)$, $\{\s_j(0)\}_{j=1}^{J}$ and
$\{\bPhi^{-1}_j(0)\}_{j=1}^{J}$ \FOR{$t=1,2,\ldots$}
   \STATE All $j\in\mathcal{V}$:
   \IF{$|x_j(t)-\h_j^T(t)\s_j(t-1)|\leq \tau\sigma_j(t)$}
   \STATE stay idle
   \ELSE
   \STATE set $\s_{j'}(t-1)$ as recently received
   ones from all $j' \in \N_j$
   \STATE update $\bPhi^{-1}_j(t)$ using \eqref{eq:d-rls-phi}
   \STATE update $\s_j(t)$ using \eqref{eq:new-d-rls-s}
   \STATE transmit $\s_j(t)$ to and receive $\s_{j'}(t)$ from all $j' \in \N_j$
   \STATE compute $\bdelta_j(t)$ using \eqref{eq:new-d-rls-delta}
   \ENDIF
   \IF{do not receive from any $j' \in \N_j$ for $d_{\max}$ time}
   \STATE receive $\s_{j'}(t)$
   \ENDIF
\ENDFOR
\end{algorithmic}}
\end{algorithm}

\section{Development and performance analysis}\label{sec:theoretical}

This section starts with a criterion-based development of CD-RLS-1. Convergence analysis of all three
censoring strategies will follow, before developing practical means of setting the censoring threshold $\tau \sigma_j(t)$.

\subsection{Derivation of censoring-based D-RLS-1}
\label{sec:derivation}

Consider the following truncated
quadratic cost that is similar to the one used in the
censoring-based but centralized RLS
\cite{Berberidis2016-Censoring}
\begin{align}\label{eq:trun-quad}
& f_{j,t}(\s):= \\
& \hspace{-0.5em} \begin{cases}
0,&|x_j(t)-\h_j^T(t)\s| \leq \tau\sigma_j(t) \\
\frac{1}{2}[x_j(t)-\h_j^T(t)\s]^2-\frac{1}{2}\tau^2\sigma_j(t)^2,&|x_j(t)-\h_j^T(t)\s|
> \tau\sigma_j(t) \nonumber
\end{cases}
\end{align}
which is convex, but non-differentiable on
$\{\s:|x_j(t)-\h_j^T(t)\s| = \tau\sigma_j(t)\}$. Using
\eqref{eq:trun-quad} to replace the quadratic loss
$[x_j(\tau)-\h_j^T(\tau)\s]^2$ in \eqref{eq:ewls}, our CD-RLS-1
criterion is
\begin{align}\label{eq:ewls-new}
\min\limits_{\s}\sum_{r=1}^{t}\sum_{j=1}^{J}\lambda^{t-r}f_{j,r}(\s).
\end{align}

To solve \eqref{eq:ewls-new} in a decentralized manner, we
introduce a local estimate $\s_j$ per node $j$, along with
auxiliary vectors $\barz_j^{j'}$ and $\tildez_j^{j'}$ per edge
$(j,j')$. By constraining all local estimates of neighbors
to consent, we arrive at the following equivalent separable
convex program per slot $t$
\begin{align}\label{eq:problem}
\min\limits_{\{{\s}_j\}_{j \in \mathcal{V}}} \quad & \sum_{r=1}^t\sum_{j=1}^{J} \lambda^{t-r}f_{j,r}(\s_j)\\
s.t. \quad & \s_j= \barz_j^{j'},\s_{j'}=\tildez_j^{j'},
\barz_j^{j'}=\tildez_j^{j'},j\in \mathcal{V},j'\in \N_j. \nonumber
\end{align}

Next, we employ alternating minimization and the
stochastic Newton iteration to derive our first censoring-based
solver of \eqref{eq:problem}. To this end, consider the
Lagrangian of \eqref{eq:problem} that is given by
\begin{align}\label{eq:lagrangian}
&\La(\s,\z,\bv,\bu) = \sum_{j \in \mathcal{V}}\sum_{r=1}^{t}\lambda^{t-r}f_{j,r}(\s_j) \nonumber \\
&+\sum_{j=1}^{J}\sum_{j'\in\N_j}\big[(\bv_j^{j'})^T(\s_j-\barz_j^{j'})+(\bu_j^{j'})^T(\s_{j'}-\tildez_j^{j'})\big]
\end{align}
where $\s:=\{\s_j\}_{j \in \mathcal{V}}$ and
$\z:=\{\barz_{j}^{j'},\tildez_{j}^{j'}\}_{j\in\mathcal{V}}^{j'\in\N_j}$
are primal variables, while $\bv:=\{\bv_j^{j'} \in
\mathbb{R}^p\}_{j\in\mathcal{V}}^{j'\in\N_j}$ and
$\bu:=\{\bu_j^{j'} \in
\mathbb{R}^p\}_{j\in\mathcal{V}}^{j'\in\N_j}$ are dual variables.
Consider also the augmented Lagrangian of \eqref{eq:problem}, namely
\begin{align}\label{eq:aug-lagrangrian}
&\La_\rho(\s,\z,\bv,\bu) = \La(\s,\z,\bu,\bv) \nonumber\\
& \hspace{2em}
+\frac{\rho}{2}\sum_{j=1}^{J}\sum_{j'\in\N_j}\big[||\s_j-\barz_j^{j'}||^2+||\s_{j'}-\tildez_j^{j'}||^2\big]
\end{align}
where $\rho$ is a positive regularization scale. Note that the
constraints on $\z$ are not dualized, but they are collected in the set
$\mathcal{C}_{\z} :=
\{\z|\barz_j^{j'}=\tildez_{j}^{j'},j\in\mathcal{V},j'\in\N_j,j\neq
j'\}$.

To minimize \eqref{eq:problem} per slot $t>0$, we rely on
alternating minimization \cite{Tseng1991} in an online
manner, which entails an iterative procedure consisting of three
steps.

\vspace{1em}

\textbf{[S1] \quad Local estimate updates:}
\begin{align*}
\s(t) =
\text{arg}\min\limits_{\s}\La(\s,\z(t-1),\bv(t-1),\bu(t-1))
\end{align*}

\textbf{[S2] \quad Auxiliary variable updates:}
\begin{align*}
\z(t) = \text{arg}\min\limits_{\z\in
\mathcal{C}_{\z}}\La_\rho(\s(t),\z,\bv(t-1),\bu(t-1))
\end{align*}

\textbf{[S3] \quad Multiplier updates:}
\begin{align*}
\bv_j^{j'}(t) &= \bv_j^{j'}(t-1)+\rho \big[\s_j(t)-\barz_j^{j'}(t) \big]\\
\bu_j^{j'}(t) &= \bu_j^{j'}(t-1)+\rho \big
[\s_{j'}(t)-\tildez_j^{j'}(t) \big].
\end{align*}

Observe that [S2] is a linearly constrained quadratic program,
for which if $\bv_j^{j'}(t-1) + \bu_j^{j'}(t-1) = \mathbf{0}$, we
always have
$$\s_{j'}(t) + \s_j(t)  = \tildez_j^{j'}(t) +\barz_j^{j'}(t) \quad \text{and} \quad \tildez_j^{j'}(t) = \barz_j^{j'}(t).$$
Therefore, the initial values of $\bv_j^{j'}$
and $\bu_j^{j'}$ in [S3] are selected to satisfy $\bv_j^{j'}(0) +
\bu_j^{j'}(0) = \mathbf{0}$ (the simplest choice is $\bv_j^{j'}(0)
= \bu_j^{j'}(0) = \mathbf{0}$). It then holds for $t \geq 0$ that
$$\bv_j^{j'}(t) + \bu_j^{j'}(t) = \mathbf{0}.$$
Using the latter to eliminate $\bu_j^{j'}$ in [S3], we
obtain
\begin{align}\label{eq:v-update}
\bv_j^{j'}(t) & = \bv_j^{j'}(t-1) + \frac{\rho}{2} \big[\s_j(t)-\barz_j^{j'}(t) - \s_{j'}(t)+\tildez_j^{j'}(t)\big] \nonumber \\
              & = \bv_j^{j'}(t-1) + \frac{\rho}{2} \big[\s_j(t)-\s_{j'}(t)\big]
\end{align}
where the first equality comes from subtracting the two lines in
[S3], and the second equality is due to $\tildez_j^{j'}(t) =
\barz_j^{j'}(t)$. The auxiliary variables
$\tildez_j^{j'}$ and $\barz_j^{j'}$ can be also eliminated. When
$\bv_j^{j'}$ is initialized by $\bv_j^{j'}(0) = \mathbf{0}$,
summing up both sides of \eqref{eq:v-update} from $r=1$ to $r=t$,
we arrive, after telescopic cancellation, at
\begin{align}\label{eq:v-update-new}
\bv_j^{j'}(t) = \frac{\rho}{2} \sum_{r=1}^t
\big[\s_j(r)-\s_{j'}(r)\big].
\end{align}

Moving on to [S1], observe that it can be split into $J$ per-node subproblems
\begin{align*}
\s_j(t) & = \text{arg}\min\limits_{\s_j}\sum_{r=1}^{t}\lambda^{t-r}f_{j,r}(\s_j) \\
&\hspace{4.2em}
+\sum_{j'\in\N_j}[\bv_j^{j'}(t-1)-\bv_{j'}^{j}(t-1)]^T\s_j.
\end{align*}
Before solving \eqref{eq:trun-quad} with the stochastic Newton
iteration \cite{Amari1998}, eliminate $\bv_j^{j'}$ using
\eqref{eq:v-update-new} to obtain
\begin{align*}
\s_j(t) & = \text{arg}\min\limits_{\s_j}\sum_{r=1}^{t}\lambda^{t-r}f_{j,r}(\s_j) \\
&\hspace{4.2em} + \rho \sum_{r=1}^{t-1} \sum_{j'\in\N_j}
\big[\s_j(r)-\s_{j'}(r)\big]^T\s_j
\end{align*}
which after manipulating the double sum yields
\begin{align*}
& \s_j(t) = \text{arg}\min\limits_{\s_j}\sum_{r=1}^{t}\lambda^{t-r}f_{j,r}(\s_j) \\
& + \sum_{r=1}^{t} \lambda^{t-r} \rho \sum_{j'\in\N_j} \Big[\s_j(r-1)-\s_{j'}(r-1) \\
& \hspace{3em}+ (1-\lambda) \sum_{\xi=1}^{r-1} \big(
\s_j(\xi-1)-\s_{j'}(\xi-1) \big) \Big]^T\s_j.
\end{align*}
If the update in \eqref{eq:new-d-rls-delta} is initialized with
$\bdelta_j(0) = \mathbf{0}$, summing up both sides from $\xi=1$
to $\xi=r-1$, we find after telescopic cancellation
\begin{align}\label{eq:delta-blabla}
& \hspace{-0.3em}\bdelta_j(r-1) = \sum_{j'\in\N_j} \Big[\s_j(r-1)-\s_{j'}(r-1) \nonumber \\
& \hspace{5em} + (1-\lambda) \sum_{\xi=1}^{r-1} \big(
\s_j(\xi-1)-\s_{j'}(\xi-1) \big) \Big].
\end{align}

Thus, optimization of $\s_j(t)$ reduces to
\begin{align}\label{eq:newton-problem}
& \s_j(t) =
\text{arg}\min\limits_{\s_j}\sum_{r=1}^{t}\lambda^{t-r} g_{j,r}(\s_j)
\end{align}
where the instantaneous cost per slot $t$ is
\begin{align}\label{eq:newton-cost}
& g_{j,t}(\s_j): =f_{j,t}(\s_j)+ \rho \bdelta_j^T(t-1) \s_j.
\end{align}

The stochastic gradient of the latter is given by
\begin{align*}
& \nabla g_{j,t}(\s_j(t-1)) \\
= & - c_j(t) \Big[ \big( x_j(t) - \h_j(t) \s_j(t-1) \big) \h_j(t)
\Big] + \rho \bdelta_j(t-1).
\end{align*}
In the stochastic Newton method, the Hessian matrix is given
by
$$\M_j(t)=E[\nabla^2g_{j,t}(\s_j(t-1))]=E[c_j(t)\h_j(t)\h_j^T(t)]$$
where the second equality comes from
\eqref{eq:trun-quad} and \eqref{eq:censor-indicator}.
A reasonable approximation of the expectation is provided by
sample averaging. However, presence of $\lambda\neq 1$
affects attenuation of regressors, which leads to
\begin{align*}
\M_j(t) &=\frac{1}{t}\sum_{r=1}^{t}\lambda^{t-r}c_j(r)\h_j(r)\h_j^T(r)\\
&=
\lambda\frac{t-1}{t}\M_j(t-1)+\frac{1}{t}c_j(t)\h_j(t)\h_j^T(t).
\vspace{-0.5em}
\end{align*}
Applying the matrix inversion lemma, we obtain
\begin{align}\label{eq:inverse-M}
\M_j^{-1}(t) & = \frac{t}{t-1}\Big[\lambda^{-1}\M_j^{-1}(t-1) \\
&
-c_j(t)\frac{\lambda^{-1}\M_j^{-1}(t-1)\h_j(t)\h_j^T(t)\M_j^{-1}(t-1)}{(t-1)\lambda-\h^T(t)\M_j^{-1}(t-1)\h_j(t)}\Big]
\nonumber
\end{align}
and after adopting a diminishing step size $1/t$, the
stochastic Newton update becomes
\begin{align*}
\s_j(t) &= \s_j(t-1) - \frac{1}{t} \M_j^{-1}(t) \nabla
g_{j,t}(\s_j(t-1)).
\end{align*}
For rational convenience, let $\bPhi_j^{-1}(t) := \M_j^{-1}(t)/t$, and
rewrite \eqref{eq:inverse-M} as (cf. \eqref{eq:d-rls-phi})
\begin{align}\label{eq:phi-decreasing}
\bPhi_j^{-1}(t) & =  \lambda^{-1}\bPhi_j^{-1}(t-1) \\
&
-c_j(t)\frac{\lambda^{-1}\bPhi_j^{-1}(t-1)\h_j(t)\h_j^T(t)\bPhi_j^{-1}(t-1)}{\lambda+\h_j^T(t)\bPhi_j^{-1}(t-1)\h_j(t)}.
\nonumber
\end{align}
Substituting $\nabla g_{j,t}(\s_j(t-1))$ and $\bPhi_j^{-1}(t)$
into the stochastic Netwon iteration yields (cf.
\eqref{eq:new-d-rls-s})
\begin{align*}
\hspace{-0.5em} \s_j(t) & = \s_j(t-1)+c_j(t)\bPhi_j^{-1}(t)\h_j(t)\big[x_j(t)-\h_j^T(t)\s_j(t-1)\big]\\
& -\rho\bPhi_j^{-1}(t)\bdelta_j(t-1)
\end{align*}
which completes the development of CD-RLS-1.

\subsection{Convergence analysis}\label{sec:convergence}

Here we establish convergence of all three novel strategies
for $\lambda = 1$. With $\lambda < 1$, the EWLS estimator
can even adapt to time-varying parameter vectors, but analyzing
its tracking performance goes beyond the scope of this paper.
For the time-invariant case ($\lambda=1$), we will rely on
the following assumption.
%

%

\noindent (\textbf{as1}) \emph{Observations obey the linear model
$x_j(t)=\h_j(t)\s_0+\epsilon_j(t)$, where
$\epsilon_j(t)\sim\N(0,\sigma_j^2)$ is correlated across $j$ and
$t$. Rows $\h_j^T(t)$ are uniformly bounded and independent of
$\epsilon_j(t)$. Covariance matrices $\Rhj := E[\h_j(t)\h_j^T(t)]
\succ\mathbf{0}_{p\times p}$ are time-invariant and positive
definite. Process $\{c_j(t)\h_j(t)\h_j^T(t)\}$ is mean ergodic,
while $\{\epsilon_j(t)\}$ and $\{c_j(t)\}$ are uncorrelated.
Eigenvalues of $\bPhi_j(t)/t$, which approximate the true
positive definite Hessian matrices $E[c_j(t)\h_j(t)\h_j^T(t)]$,
are bounded below by a positive constant when $t$ is large
enough.}

We will assess convergence of our iterative algorithms using
the squared mean-root deviation (SMRD) metric, defined as
\begin{align}\label{SMRD}
\text{SMRD}(t) &:= \left\{ E \Big[ \Big(\sum_{j=1}^{J}
||\s_j(t)-\s_0||^2 \Big)^{\frac{1}{2}} \Big] \right\}^2.
\end{align}
Letting $\be_j(t) := \s_j(t)-\s_0 \in \mathbb{R}^p$ denote the
estimation error of node $j$ and $\be(t):=[\be_1^T(t), \ldots,
\be_J^T(t)]^T \in \mathbb{R}^{Jp}$ the estimation error across all
nodes, one can see that $\text{SMRD}(t) =
\left\{E[\|\be(t)\|]\right\}^2$. Observe that $\text{SMRD}(t)$ is
a lower-bound approximation of the mean-square deviation (MSD)
metric $\text{MSD}(t) := E[\|\be(t)\|^2]$
\cite{Global-performance,MSD2}, since by Jensen's inequality
$\left\{E[\|\be(t)\|]\right\}^2 \leq E[\|\be(t)\|^2]$.


Under (as1), convergence of CD-RLS-1 and CD-RLS-2 is asserted as
follows; see Appendix \ref{sec:proof-1} for the proof.

\begin{theorem}\label{thm:convergence} For CD-RLS-1 and
CD-RLS-2 Algorithms 1 and 2, set $\sigma_j(t)=\sigma_j$ and
$\bPhi^{-1}_j(0) = \gamma \mathbf{I}_p$ per node $j$. Let $\mu :=
\min\{\lambda_{\min}(\Rhj), j \in \mathcal{V}\}$, and suppose $0 <
\rho < 1 / (\gamma\lambda_{\max}(\bL))$ for CD-RLS-1 and
correspondingly $0 < \rho < \rho_0 $ for CD-RLS-2, while $\bL$ is
the network Laplacian and the constant $\rho_0$ depends on
$\lambda_{\max}(\bL), \gamma, \tau, \mu$, and the upper bound of
$\h_j(t)$. Under (as1), there exists $t_0
> 0$ for which it holds for $t > t_0$ that
\begin{align}
     & \left\{ E \Big[ \Big(\sum_{j=1}^{J}
||\s_j(t)-\s_0||^2 \Big)^{\frac{1}{2}} \Big] \right\}^2 \nonumber \\
\leq & \sum_{j=1}^{J}\frac{\gamma^{-1}||\s_j(0)-\s_0||^2 +
\gamma t_0\sigma_j^2\mathrm{tr}(\Rhj)}{2Q(\tau)\mu t}\nonumber\\
+&  \frac{\gamma
\sigma_j^2\lambda_{\max}(\Rhj^{-1})\mathrm{tr}(\Rhj)\ln(t)}{4Q^2(\tau)\mu
t}. \label{eq:main}
\end{align}
\end{theorem}

Theorem \ref{thm:convergence} establishes that the SMRD in
\eqref{SMRD} converges to zero at a rate $O(\ln(t)/t)$. The
constant of the convergence rate is related to $\Rhj$ through
$\lambda_{\max}(\Rhj^{-1})$, $\mathrm{tr}(\Rhj)$ and $\mu$; the
noise covariance $\sigma_j^2$, and the threshold $\tau$ through
$Q(\tau)$. Theorem \ref{thm:convergence} also indicates the impact
of the initial states (determined by $\gamma$ and $\s_j(0)$),
which disappears at a faster rate of $O(1/t)$. To guarantee
convergence, the step size $\rho$ must be small enough.


The proof for CD-RLS-3 is more challenging. Because a node does
not receive any information from its neighbors when censoring
is in effect, it has to rely on outdated neighboring estimates when the
incoming datum is not censored. This delay in percolating information
may cause computational instability. For this reason, we will impose
an additional constraint to guarantee that all local estimates
do not grow unbounded. In practice, this can be realized by truncating
local estimates when they exceed a certain threshold.

\noindent
(\textbf{as2}) \emph{ Local estimates $\{\s_j(t)\}_{j=1}^J$
    are uniformly bounded $\forall t\geq 0$.}

Convergence of CD-RLS-3 is then asserted as follows. Similar to
CD-RLS-1 and CD-RLS-2, the SMRD of CD-RLS-3 converges to zero with
rate $O(\ln(t)/t)$, as stated in the following theorem.


%
\begin{theorem}\label{thm:convergence3}
For CD-RLS-3 given by Algorithms 3, set
$\sigma_j(t)=\sigma_j$ and $\bPhi^{-1}_j(-1) = \gamma
\mathbf{I}_p$ per node $j$. Under (as1) and (as2) with $0 < \rho <
\rho_0$ as in Theorem \ref{thm:convergence}, there exists $t_0>0$
for which it holds $\forall t > t_0$, that
\begin{align}
     & \left\{ E \Big[ \Big(\sum_{j=1}^{J}
||\s_j(t)-\s_0||^2 \Big)^{\frac{1}{2}} \Big] \right\}^2 \leq
\frac{a+b\ln(t)}{t} \label{eq:extra}
\end{align}
where $a$ and $b$ are positive constants that depend on the upper
bounds of $\h_j(t)$ and $\s_j(t)$, parameters $\rho$ and $\tau$,
the covariance $\Rhj(t)$, the Laplacian matrix $\bL$, and $t_0$.
\end{theorem}

Although the bounds asserted by Theorems
\ref{thm:convergence} and \ref{thm:convergence3} could be loose,
they demonstrate that $\lim\sup_{t \rightarrow \infty}
\text{SMRD}(t) = 0$, which establishes that the decentralized
estimates converge to the ground truth asymptotically.

\subsection{Threshold setting and variance estimation}\label{sec:threshold}

The threshold $\tau$ influences considerably the performance of
all CD-RLS algorithms. Its value trades off estimation accuracy for
computation and communication overhead. We provide a simple criterion for
setting $\tau$ using the average censoring ratio $\pi^{*}$, which is
defined as the number of censored data over the total number of
data \cite{Msechu2012-cratio}. The goal is to choose $\tau$ so
that the actual censoring ratio approaches $\pi^{*}$ as $t$ goes
to infinity -- since we are dealing with streaming big data, such
an asymptotic property is certainly relevant. When $t$ is large
enough, $\s$ is very close to $\s_0$; thus, the innovation
$x_j(t)-\h_j^T(t)\s_j(t-1)\approx
x_j(t)-\h_j^T(t)\s_0=\epsilon_j(t)\sim\N(0,\sigma_j^2)$. As a
consequence, $\Pr(c_j(t)=0) = \Pr(|x_j(t)-\h_j^T(t)\s_j(t-1)|\leq
\tau\sigma_j)\approx \Pr(|\epsilon_j(t)|\leq
\tau\sigma_j)=\Pr(|\epsilon_j(t)/\sigma_j|\leq \tau)=1-2Q(\tau)$,
where the last equality holds because
$\epsilon_j(t)/\sigma_j\sim\N(0,1)$. Therefore, $\pi^{*} = \lim_{t
\rightarrow \infty}
\frac{1}{t}\sum_{\tau=0}^{t}E[c_j(\tau)]\approx 1-2Q(\tau)$, which
implies that
\[
\tau = Q^{-1}((1-\pi^{*})/2)\;.
\]

Given the average censoring ratio $\pi^{*}$, Table \ref{tab:comp}
compares the average per step per node communication and
computational costs of D-RLS and the proposed CD-RLS algorithms.
We assume that transmitting or receiving a $p$-dimensional local
estimate vector to or from a neighboring node incurs a cost of
$p$. Thus, for D-RLS and CD-RLS-1, the average communication costs
are both $2p |\mathcal{E}|/J$. In CD-RLS-2, a node does not
transmit to its neighbors when it censors a datum, which leads to
an average communication cost of $2p |\mathcal{E}|(1-\pi^*)/J$.
CD-RLS-3 avoids communication over a link as long as one of the
two end nodes censors a datum, and hence reduces the cost to $2p
|\mathcal{E}|(1-\pi^*)^2/J$. As discussed in Section
\ref{sec:c-d-rls}, the computational costs of CD-RLS-1 for the
non-censoring and censoring cases are $O(7p^2/2)$ and $O(p^2)$,
respectively. For the censoring case, CD-RLS-2 and CD-RLS-3 reduce
their computational costs to $O(p)$, and are more computationally
efficient.

\begin{table}\addtolength{\tabcolsep}{-0pt} \centering
\caption{Average per step per node communication and computational
costs, given the average censoring ratio
$\pi^{*}$.}\label{tab:comp} \vspace{-0.3cm}
\begin{tabular}{ c*{2}{|c}} \hline
Algorithm   & Communication & Computation \\
\hline
D-RLS    & $2p |\mathcal{E}|/J$            & $7p^2/2 + O(p)$ \\
CD-RLS-1 & $2p |\mathcal{E}|/J$            & $7p^2 (1-\pi^*)/2 + p^2 \pi^* + O(p) $ \\
CD-RLS-2 & $2p |\mathcal{E}|(1-\pi^*)/J$   & $7p^2 (1-\pi^*)/2 + O(p)$ \\
CD-RLS-3 & $2p |\mathcal{E}|(1-\pi^*)^2/J$ & $7p^2 (1-\pi^*)/2 + O(p)$ \\
\hline
\end{tabular} \vspace{-0.3cm}
\end{table}

If the variances $\{\sigma_j^2\}$ were known, one could simply
choose $\sigma_j(t)=\sigma_j$. However, $\sigma_j$ in practice is
often unknown. In this case, we consider the running average
$\sigma_j^2(t+1) \approx
t^{-1}\sum_{\tau=1}^{t+1}[x_j(\tau)-\h_j^T(\tau)\s_0]^2 =
(t-1)\sigma_j^2(t)/t+[x_j(t+1)-\h_j^T(t+1)\s_0]^2/t$, which
suggests the recursive variance estimate
\[
\sigma_j^2(t+1) = (t-1)\sigma_j^2(t)/t+[x_j(t+1)-\h_j^T(t+1)\s_j(t)]^2/t\;.
\]

\begin{figure}
    \centering
    \includegraphics[height=6.2cm] {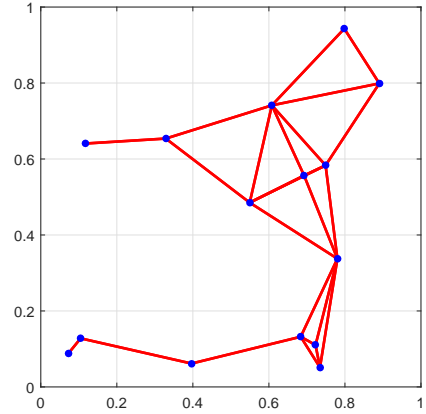}
    \caption{The network topology used in the numerical experiments.}
    \label{eps:network}
\end{figure}

\section{Numerical Experiments}
\label{sec:numerical}

This section provides numerical results to validate the
effectiveness of our novel censoring strategies. We simulate a
network of $J=15$ nodes, which are uniformly randomly deployed
over a $1 \times 1$ square. Two nodes within communication range
$0.3$ are deemed as being neighbors. The resultant network
topology is depicted in Fig. \ref{eps:network}. We compare six
algorithms: the centralized adaptive censoring (AC)-RLS that runs
in every node independently, the distributed diffusion least
mean-square (Diffusion-LMS) algorithm \cite{Bianchi2013,
Mateos2009-JASP}, D-RLS without censoring \cite{Mateos2009-DRLS},
and the three censoring-based D-RLS algorithms, namely CD-RLS-1,
CD-RLS-2 and CD-RLS-3. All algorithms are evaluated on two data
sets, one synthetic and one real. The empirical SMRD is used as
performance metric.

For the synthetic data set, the unknown $\s_0$ is $p$-dimensional
with $p=4$. The setting is the one in \cite{Mateos2009-DRLS},
where WSN-based decentralized power spectrum estimation is sought
for a signal modeled as an autoregressive process. In this
context, consider an auxiliary sequence $r_j(t)$ that evolves
according to $r_j(t) = (1-q)\beta_j r_j(t-1)+\sqrt{q}\omega_j(t)$.
Starting from $r_j(t)$, the row $\h_j^T(t)$ is formed by taking
the next $p$ observations, namely $\h_j^T(t) = [r_j(t+p-1);
\ldots; r_j(t)]$. Parameters are selected as $q=0.5$,
$\beta_j\sim\mathcal{U}(0,1)$, and also uniformly distributed
driving white noise $\omega_j(t)\sim\mathcal{U}(-\sqrt
3\sigma_{\omega_j},\sqrt 3\sigma_{w_j})$ with
$\sigma^2_{\omega_j}\sim\mathcal{U}(0,2)$. Observation of node $j$
is subject to additive white Gaussian noise, with covariance
$\sigma_j^2 = 10^{-3} \alpha_j$, where
$\alpha_j\sim\mathcal{U}(0,1)$. The true signal vector is
$\mathbf{s}_0 = \mathbf{1}_p$, for which $\lambda=1$ is set for
all algorithms. For D-RLS, CD-RLS-1, CD-RLS-2 and CD-RLS-3, the
step size $\rho=0.01$ and $\bPhi_j^{-1}(0)=\gamma \mathbf{I}_p$
where $\gamma = 30$, leading to fastest convergence of D-RLS.
Regarding the four censoring-based algorithms AC-RLS, CD-RLS-1,
CD-RLS-2 and CD-RLS-3, we set the average censoring ratio to
$\pi^{*} = 0.6$, which is approached using $\tau =
Q^{-1}((1-\pi^{*})/2)\approx 0.84$. The variances $\sigma_j^2$ are
estimated in an online manner as described in Section
\ref{sec:threshold}. AC-RLS uses $\bPhi_j^{-1}(0)=\gamma
\mathbf{I}_p$, where $\gamma = 10^5$ leads to the fastest
convergence. Diffusion-LMS uses the nearest-neighbor diffusion
matrix and $1.5/\sqrt{t}$ step size, which is tuned to obtain
fastest convergence. For all curves obtained by running the
algorithms, the ensemble averages are approximated via sample
averaging over 100 Monte Carlo runs.

\begin{figure}
    \centering
    \includegraphics[height=6.2cm] {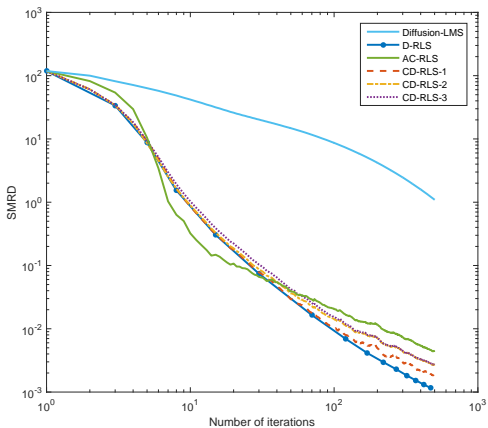}
    \caption{SMRD of the six algorithms versus number of iterations.}
    \label{eps:iteration}
\end{figure}

\begin{figure}
    \centering
    \includegraphics[height=6.2cm] {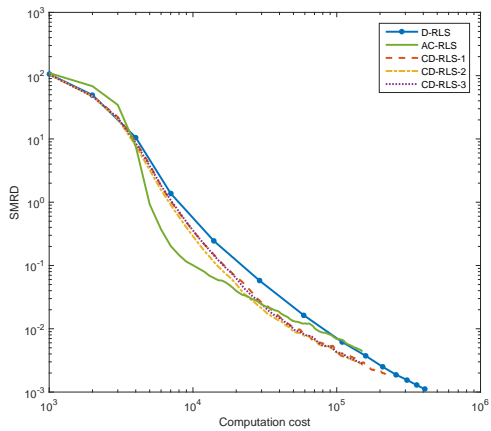}
    \caption{SMRD of the five algorithms versus computational cost, defined as the number of multiplications.}
    \label{eps:computation}
\end{figure}

\begin{figure}
    \centering
    \includegraphics[height=6.2cm] {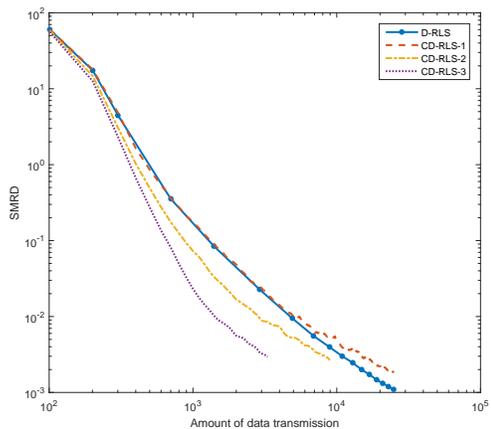}
    \caption{SMRD of the four decentralized algorithms versus amount of data transmission in the unicast mode. }
    \label{eps:transmission}
\end{figure}

Fig. \ref{eps:iteration} depicts the SMRD versus the number of
iterations. Not surprisingly, since D-RLS does not censor data,
its convergence rate with respect to the number of iterations is
the fastest. Among the three proposed CD-RLS algorithms, CD-RLS-2
and CD-RLS-3 are slower than CD-RLS-1, because the former two
incur smaller communication cost than the latter. Though CD-RLS-3
adopts a more aggressive censoring strategy than CD-RLS-2, its
convergence does not degrade as confirmed by Fig.
\ref{eps:iteration}. AC-RLS is the slowest among all except for
Diffusion-LMS, because it is run at all nodes independently,
without sharing information over the network. Even though the SMRD
of Diffusion-LMS vanishes as $t\rightarrow \infty$ (with rate
$1/t$), its finite-sample SMRD decays slower than our CD-RLS
schemes for which SMRD also vanishes as $t\rightarrow \infty$
(with rate upper bounded by $\ln(t)/t$). This is analogous to
\emph{centralized} LMS that for finite samples exhibits SMRD
decaying slower than that of \emph{centralized} RLS. Note that contrary to the analysis in \cite{Bianchi2013} and \cite{Morral2017}, the cost function here is not differentiable and thus the Diffusion-LMS does not achieve the traditional linear rate. We shall not
compare with Diffusion-LMS in the rest of the numerical
experiments.

The merits of censoring are further appreciated when one
considers computational costs. Recall that the target average censoring
ratio is $\pi^{*} = 0.6$, meaning that $3/5$ of the data are
discarded (actual values are $0.6320$ for AC-RLS, $0.6292$ for
CD-RLS-1, $0.6277$ for CD-RLS-2, and $0.6237$ for CD-RLS-3,
averaged over 100 runs). As confirmed by Fig.
\ref{eps:computation}, the three CD-RLS algorithms consume considerably
less computational resources relative to D-RLS that does not censor data.
Indeed, whenever a datum is censored, CD-RLS-1 only requires $2/7$
of the computations relative to D-RLS, while CD-RLS-2 and
CD-RLS-3 incur minimal computational overhead. Although AC-RLS is the most
computationally efficient algorithm at the beginning,
absence of collaboration undermines its performance in steady state.

\begin{figure}
    \centering
    \includegraphics[height=6.2cm] {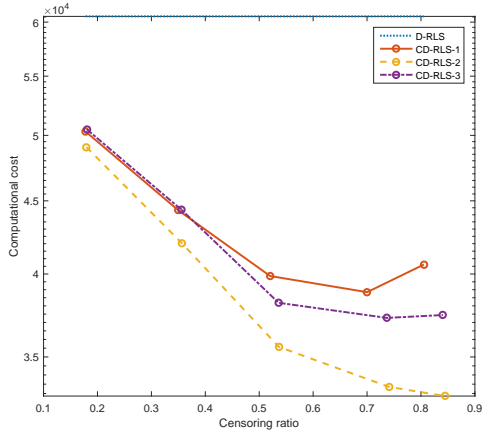}
    \caption{Computational cost of the four decentralized algorithms for variable censoring ratios when target SMRD is $0.015$.}
    \label{eps:computation-cr}
\end{figure}

\begin{figure}
    \centering
    \includegraphics[height=6.2cm] {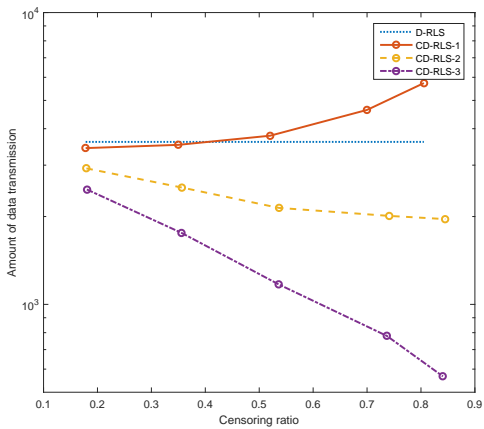}
    \caption{Amount of data transmission of the four decentralized algorithms for variable censoring ratios when target SMRD is $0.015$.}
    \label{eps:communication-cr}
\end{figure}

Regarding the amount of data exchanged to communicate local
estimates in a unicast mode, CD-RLS-1 is the worst because nodes
need to transmit their local estimate to neighbors, no matter
whether local data are censored or not. Fig.
\ref{eps:transmission} corroborates that CD-RLS-2 and CD-RLS-3
show significant improvement over D-RLS, demonstrating their
potential for reducing both communication and computation costs in
solving decentralized linear regression problems over large-scale
networks.

We further numerically quantify the savings of computation and
communication that the three censoring-based D-RLS algorithms
enjoy over RLS without censoring. We set the target SMRD to
$0.015$ and plot the computational and communication costs
required to reach it. According to Fig. \ref{eps:computation-cr},
the computational costs of the three censoring-based algorithms
decrease to about half of that of D-RLS as the censoring ratio
grows to $0.7$, while CD-RLS-2 outperforms the other two. Though
CD-RLS-2 uses more iterations (hence more data) to achieve the
target SMRD than CD-RLS-1 (see Fig. \ref{eps:iteration}), it
requires less computation when a datum is censored. On the other
hand, CD-RLS-3 uses more iterations to achieve the target SMRD
than CD-RLS-2, and hence it incurs more computational cost. The
saving of CD-RLS-3 over CD-RLS-2 is mainly in the communication
cost. In Fig. \ref{eps:communication-cr}, the communication cost
of CD-RLS-2 and CD-RLS-3 decreases as the censoring ratio grows,
but that of CD-RLS-1 increases and is larger than that of D-RLS
when the censoring ratio exceeds $0.5$. CD-RLS-3 exhibits best
performance in terms of communication cost.

Next, we vary $\pi$ and evaluate its impact on SMRD, as shown in
Fig. \ref{eps:smrd-cr}. The SMRD here is computed after 500
iterations. When $\pi$ is close to $0.5$, meaning about $1/2$ of
the data is censored, the three proposed CD-RLS algorithms are
still able to reach SMRD of $10^{-4}$, which is the limit of D-RLS
without censoring. Among the three algorithms, CD-RLS-1 exhibits
the best SMRD curve, but its computation and communication costs
are the highest. AC-RLS does not perform well especially for low
censoring ratios due to the lack of network-wide collaboration.
CD-RLS-2 and CD-RLS-3 perform comparably in this experiment.

\begin{figure}
    \centering
    \includegraphics[height=6.2cm] {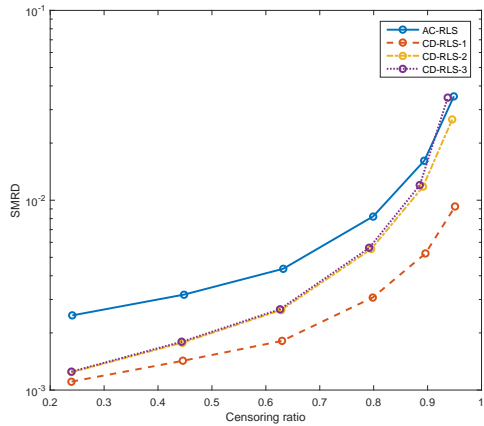}
    \caption{SMRD after 500 iterations of the four censoring algorithms for variable censoring ratios.}
    \label{eps:smrd-cr}
\end{figure}

\begin{figure}
    \centering
    \includegraphics[height=6.2cm] {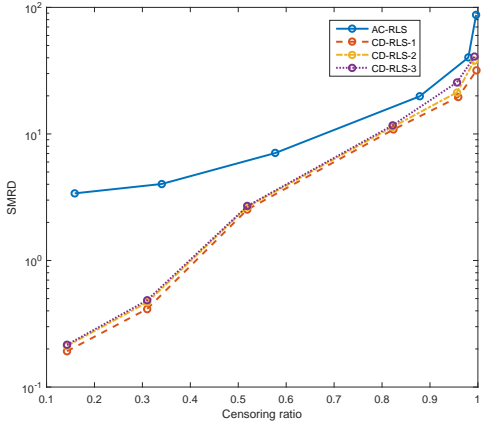}
    \caption{SMRD of the four censoring algorithms versus the censoring ratio on a real data set of protein tertiary structures.}
    \label{eps:smrd-real}
\end{figure}

The effectiveness of the novel censoring-based strategies is
further assessed on a real data set of protein tertiary structures
\cite{Realdata}. The premise here is that a given dataset is not
available at a single location, but it is distributed over a
network whose nodes are interested in obtaining accurate
regression coefficients while suppressing the communication and
computational overhead. Again, the graph in Fig. \ref{eps:network}
is used to model the network of regression-performing agents. The
number of control variables is $p=9$. The first $45,720$ (out of
$45,730$) observations are normalized and divided evenly into
$J=15$ parts, one per node. For CD-RLS-1, CD-RLS-2 and CD-RLS-3,
we set $\rho=0.05$ and $\bPhi_j^{-1}(0)=5 \mathbf{I}_p$, while for
AC-RLS we choose $\gamma=10$. The ground truth vector $\s_0$ is
estimated by solving a batch least-squares problem on the entire
data set. Similar to what we deduced from Fig. \ref{eps:smrd-cr} in
the synthetic data set, the novel CD-RLS algorithms outperform
AC-RLS in terms of SMRD, as one varies the average censoring ratio
from $15\%$ to nearly $100\%$ in Fig. \ref{eps:smrd-real}.

\begin{figure}
    \centering
    \includegraphics[height=6.2cm] {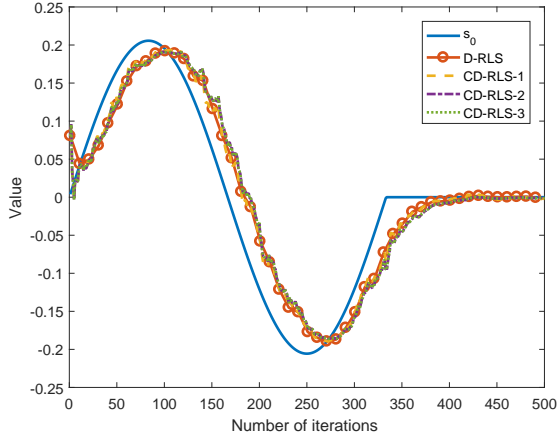}
    \caption{First entries in the vector estimates of the four algorithms versus number of iterations when $\lambda = 0.95$.}
    \label{eps:lambda}
\end{figure}

When $\lambda < 1$, the three censoring-based strategies are also
able to track time-varying signals well. Note that to track the
signal dynamics in this case, the censoring ratio cannot be too
large. We use the same setting of the synthetic data but change
the true $\s_0$ such that its $i$th element is
$\tilde{\beta}_i\sin(3 \pi t/500)$ when $t \leq 1000/3$, and
remains constant after $t = 1000/3$. The magnitudes
$\tilde{\beta}_i$ are i.i.d. and follow $\mathcal{U}(0,1)$. The
parameters of the four decentralized algorithms are the same as
those in the previous synthetic experiments, except that the
censoring ratio is $0.3$ when the censoring strategies are
applied. Fig. \ref{eps:lambda} depicts the evolution of the first
entries in the vector estimates of the four algorithms. They show
similar tracking performance, but the censoring-based algorithms
incur lower communication and computation costs over D-RLS.

\section{Concluding Remarks} 
\label{sec:concluding_remarks}

This paper introduced three data-adaptive censoring strategies that
significantly reduce the computation and communication costs of
the RLS algorithm over large-scale networks. The basic idea
behind these strategies is to avoid inefficient computation and
communication when the local observations and/or the neighboring
messages are not informative. We proved
convergence of the resulting algorithms in the mean-square deviation sense. Numerical
experiments validated the merits of the novel schemes.

The notion of identifying and discarding less informative
observations can be widely used in various large-scale online
machine learning tasks including nonlinear regression, matrix
completion, clustering and classification, to name a few. These constitute our future research
directions.

\appendices

\section{Equivalent Form of D-RLS}
\label{sec:eq-DRLS}
%
%
Here we prove that D-RLS recursions \eqref{eq:d-rls-phi} - \eqref{eq:d-rls-v} are
equivalent to \eqref{eq:d-rls-phi}, \eqref{eq:new-d-rls-s} and
\eqref{eq:new-d-rls-delta}. It follows from \eqref{eq:d-rls-s} that
\begin{align} \label{eq:app-a-001}
& \bPhi_j(t)\s_j(t)-\lambda\bPhi_j(t-1)\s_j(t-1)\\
&= \Big[\bpsi_j(t)-\frac{1}{2}\sum_{j'\in \N_j}(\bv_j^{j'}(t-1)-\bv_{j'}^{j}(t-1)) \Big] \nonumber \\
&-\lambda \Big[\bpsi_j(t-1)-\frac{1}{2}\sum_{j'\in
\N_j}(\bv_j^{j'}(t-2)-\bv_{j'}^{j}(t-2)) \Big]. \nonumber
\end{align}
Applying the matrix inversion lemma to \eqref{eq:d-rls-phi} yields
\begin{align} \label{eq:app-a-002}
\bPhi_j(t) = \lambda \bPhi_j(t-1) + \h_j(t)\h_j^T(t).
\end{align}
Substituting $\bpsi_j(t) - \lambda \bpsi_j(t-1) = \h_j(t) x_j(t)$
from \eqref{eq:d-rls-psi} and $\lambda \bPhi_j(t-1) = \bPhi_j(t) -
\h_j(t)\h_j^T(t)$ from \eqref{eq:app-a-002} into
\eqref{eq:app-a-001}, leads to
\begin{align} \label{eq:app-a-003}
& \bPhi_j(t) \big[ \s_j(t)-\s_j(t-1) \big] = \h_j(t)\big[x_j(t)-\h_j^T(t)\s_j(t-1) \big] \nonumber \\
& \hspace{6em} - \frac{1}{2}\sum_{j'\in \N_j}(\bv_j^{j'}(t-1)
- \lambda \bv_j^{j'}(t-2)) \nonumber \\
& \hspace{6em} + \frac{1}{2}\sum_{j'\in \N_j}(\bv_{j'}^{j}(t-1) -
\lambda\bv_{j'}^{j}(t-2)).
\end{align}

Next, we will show that if $\bdelta(t)$ is defined as
\begin{align} \label{eq:app-a-004}
& \bdelta(t) := \frac{1}{2\rho}\sum_{j'\in \N_j}(\bv_j^{j'}(t) -
\lambda \bv_j^{j'}(t-1)) \nonumber \\
& \hspace{2em} - \frac{1}{2\rho}\sum_{j'\in \N_j}(\bv_{j'}^{j}(t)
- \lambda\bv_{j'}^{j}(t-1))
\end{align}
then its update is exactly \eqref{eq:new-d-rls-delta}. This can be
done by taking the difference between slots $t$ and $t-1$ for
\eqref{eq:app-a-004}, and substituting the update of $\bv_j^{j'}$
in \eqref{eq:d-rls-v}. Due to \eqref{eq:app-a-004}, it follows that
\eqref{eq:app-a-003} is equivalent to
\begin{align} \label{eq:app-a-005}
& \bPhi_j(t) \big[ \s_j(t)-\s_j(t-1) \big] =
\h_j(t)\big[x_j(t)-\h_j^T(t)\s_j(t-1) \big] \nonumber \\
& \hspace{10.3em} - \rho \bdelta(t-1).
\end{align}
Left multiplying
\eqref{eq:app-a-005} with $\bPhi_j^{-1}(t)$, yields the update of $\s_j$ in
\eqref{eq:new-d-rls-s}, and completes the proof.

\section{Proof of Theorem \ref{thm:convergence}}
\label{sec:proof-1}
%
%
\begin{proof}
We need the following lemma in \cite[Chapter 7, Theorem 4]{Grimmett2011}.
\begin{lemma}\label{lem:convergence}
Let $X,X_1,X_2,...$ be random variables on some probability space. If $X_n\rightarrow X$ in probability and $Pr(|X_n|\leq k) = 1$ for all $n$ and some $k$, then $X_n\rightarrow X$ in $r$th mean for all $r\geq 1$.
\end{lemma}

Starting with
CD-RLS-1, the proof proceeds in five stages.

\textbf{Stage 1.} We first investigate the spectral properties of
$\bPhi_j(t)$ when $t$ is sufficiently large. Letting $\lambda=1$
and applying the matrix inversion lemma to the censoring form \eqref{eq:d-rls-phi}, we have
\begin{align}\label{eq:phi-one}
\bPhi_j(t) = \bPhi_j(t-1) + c_j(t)\h_j(t)\h_j^T(t).
\end{align}
Summing up from $r=1$ to $r=t$ and using the telescopic cancellation, \eqref{eq:phi-one} yields
\begin{align}\label{eq:phi-two}
\bPhi_j(t) =
\sum_{r=1}^{t}c_j(r)\h_j(r)\h_j^T(r)+\gamma^{-1}\mathbf{I}_{p}.
\end{align}
Thanks to the strong law of large numbers, $\bPhi_j(t)/t$
converges to $E[c_j(t)\h_j(t)\h_j^T(t)]$ almost surely as $t
\rightarrow \infty$. Observe that
\begin{align} \label{eq:newnew-01}
  & E[c_j(t)\h_j(t)\h_j^T(t)] \\
= & E\big[\h_j(t)\h_j^T(t)E[c_j(t)|\h_j(t),\s_j(t-1)]\big] \nonumber \\
= & E[\h_j(t)\h_j^T(t) \Pr(c_j(t)=1|\h_j(t),\s_j(t-1))]. \nonumber \\
= & E\left[\h_j(t)\h_j^T(t) \left(
1-\int_{-\tau+\sigma_j^{-1}[\h_j^T(t)(\s_j(t-1)-\s_0)]}^{\tau+\sigma_j^{-1}[\h_j^T(t)(\s_j(t-1)-\s_0)]}\phi(x)dx
\right)\right]. \nonumber
\end{align}
Observing the integral in \eqref{eq:newnew-01}, we know that
\begin{align}\label{eq:newnew-02}
1 >  & 1 - \int_{-\tau+\sigma_j^{-1}[\h_j^T(t)(\s_j(t-1)-\s_0)]}^{\tau+\sigma_j^{-1}[\h_j^T(t)(\s_j(t-1)-\s_0)]}\phi(x)dx \nonumber \\
\geq & 1 - \int_{-\tau}^{\tau}\phi(x)dx = 2Q(\tau)
\end{align}
where the event set that the second inequality strictly holds
(namely, ``$\geq$'' becomes ``$>$'') is with nonzero measure.
Thus, substituting \eqref{eq:newnew-02} into \eqref{eq:newnew-01}
yields $$ E[c_j(t)\h_j(t)\h_j^T(t)] \prec E[\h_j(t)\h_j^T(t)] =
\Rhj $$ and $$ E[c_j(t)\h_j(t)\h_j^T(t)] \succ 2Q(\tau)
E[\h_j(t)\h_j^T(t)] = 2Q(\tau)\Rhj. $$

Since $\bPhi_j(t)/t$ converges to
$E[c_j(t)\h_j(t)\h_j^T(t)]$ almost surely as $t \rightarrow
\infty$ and $\h_j(t)$ is uniformly bounded such that
$\bPhi_j(t)/t$ is also bounded (cf. \eqref{eq:phi-two}), we have
$E[\|\bPhi_j(t)/t\|_2]$ converges to
$E[\|c_j(t)\h_j(t)\h_j^T(t)\|_2]$ as $t \rightarrow \infty$
by lemma \ref{lem:convergence}. Therefore,
$2Q(\tau)\Rhj \prec E[c_j(t)\h_j(t)\h_j^T(t)] \prec \Rhj$ implies
that there exists $t_1 > 0$, for which it holds $\forall t \geq
t_1$ that
\begin{align}
\hspace{-1em}2Q(\tau) \|\Rhj\|_2 < E[\|\bPhi_j(t)/t\|_2 ] <
\|\Rhj\|_2 \nonumber
\end{align}
and consequently the expected maximum eigenvalue of $\bPhi_j(t)$
satisfies
\begin{align}\label{eq:ineq_phi-001}
\hspace{-1em}2Q(\tau) \lambda_{\max}(\Rhj) t < E[\lambda_{\max}(
\bPhi_j(t)) ] < \lambda_{\max}(\Rhj) t.
\end{align}
Observe that $t \bPhi^{-1}_j(t)$ converges to
$\left\{E[c_j(t)\h_j(t)\h_j^T(t)]\right\}^{-1}$ almost surely as
$t \rightarrow \infty$ due to the convergence of $\bPhi_j(t)/t$ to
$E[c_j(t)\h_j(t)\h_j^T(t)]$. Since eigenvalues of $\bPhi_j(t)/t$
are bounded below by a positive constant when $t$ is large enough,
there exists $t_2
> 0$ such that $t\bPhi^{-1}_j(t)$ is bounded $\forall t \geq t_2$.
Following the same analysis to obtain \eqref{eq:ineq_phi-001}, it
holds $\forall t \geq t_2$ that
\begin{align}\label{eq:ineq_phi-002}
\hspace{-1em} \lambda_{\max}(\Rhj^{-1})/t < E[\lambda_{\max}(
\bPhi^{-1}_j(t)) ] < \lambda_{\max}(\Rhj^{-1}) / (2Q(\tau) t).
\end{align}
Letting $t_0 := \max(t_1, t_2)$, \eqref{eq:ineq_phi-001} and
\eqref{eq:ineq_phi-002} hold $\forall t \geq t_0$.

\textbf{Stage 2.} Rewrite the update of $\s_j$ as
\begin{align*}
\s_j(t) & = \s_j(t-1)+c_j(t) \bPhi_j^{-1}(t)\h_j(t)\big[x_j(t)-\h_j^T(t)\s_j(t-1)\big] \\
& -\rho\bPhi_j^{-1}(t)\bdelta_j(t-1).
\end{align*}
Note also that for $\lambda=1$, the update of $\bdelta_j$ is
equivalent to (cf. \eqref{eq:delta-blabla})
\begin{align*}
& \bdelta_j(t-1) = \sum_{j'\in\N_j} \big[\s_j(t-1)-\s_{j'}(t-1)
\big].
\end{align*}

Letting $\be_j(t) := \s_j(t)-\s_0$, the estimation error obeys the recursion
\begin{align*}
\be_j(t) &=\be_j(t-1)+c_j(t)\bPhi_j^{-1}(t)\h_j(t)[x_j(t)-\h_j^T(t)\s_j(t-1)] \\
&-\rho\bPhi_j^{-1}(t)\sum_{j'\in\N_j} \big[
\be_j(t-1)-\be_{j'}(t-1) \big].
\end{align*}
Substituting $x_j(t)=\h_j(t)\s_0+\epsilon_j(t)$ to eliminate
$\s_j(t-1)$, we obtain
\begin{align}\label{eq:conv1}
\be_j(t) & =\be_j(t-1)-c_j(t)\bPhi_j^{-1}(t)\h_j(t)\h_j^T(t)\be_j(t-1)\nonumber\\
& +c_j(t)\bPhi_j^{-1}(t)\h_j(t)\epsilon_j(t) \nonumber \\
& -\rho \bPhi_j^{-1}(t)\sum_{j'\in\N_j} \big[
\be_j(t-1)-\be_{j'}(t-1) \big].
\end{align}

Left multiplying \eqref{eq:conv1} with $\bPhi_j(t)$  yields
\begin{align}\label{eq:conv2}
  & \bPhi_j(t)\be_j(t)  \\
  = & \bPhi_j(t)\be_j(t-1)-c_j(t) \h_j(t)\h_j^T(t)\be_j(t-1)\nonumber\\
  + & c_j(t) \h_j(t)\epsilon_j(t)  -\rho \sum_{j'\in\N_j} \big[\be_j(t-1)-\be_{j'}(t-1) \big] \nonumber \\
  = & \bPhi_j(t-1)\be_j(t-1)  \nonumber\\
  + & c_j(t)\h_j(t)\epsilon_j(t)-\rho\sum_{j'\in\N_j} \big[ \be_j(t-1)-\be_{j'}(t-1)
  \big]. \nonumber
\end{align}

Our convergence analysis result will rely on a matrix form of
\eqref{eq:conv2} that accounts for all nodes $j$. Define vectors
$\be(t):=[\be_1^T(t), \ldots, \be_J^T(t)]^T \in \mathbb{R}^{Jp}$,
$\beps(t):=[\epsilon_1^T(t),\ldots,\epsilon_J^T(t)]^T \in
\mathbb{R}^J$, as well as block-diagonal matrices $\bPhi(t) :=
\mathrm{diag}(\{\bPhi_j(t)\}) \in \mathbb{R}^{Jp \times Jp}$,
$\bC(t) := \mathrm{diag}(\{c_j(t)\})\in \mathbb{R}^{J \times J}$,
and $\bH(t) := \mathrm{diag}(\{\h_j(t)\}) \in \mathbb{R}^{Jp
\times J}$. Then \eqref{eq:conv2} can be written in matrix form as
\begin{align}\label{eq:matrixform-1}
  & \bPhi(t)\be(t)\nonumber \\
= &
\big[\bPhi(t-1)-\rho\bL\otimes\mathbf{I}_p\big]\be(t-1)+\bH(t)\bC(t)\beps(t)
\end{align}
which after left multiplication with $\bPhi^{-\frac{1}{2}}(t)$ yields
\begin{align}\label{eq:matrixform-2}
  \bPhi^{\frac{1}{2}}(t)\be(t) = &
\bPhi^{-\frac{1}{2}}(t)\big[\bPhi(t-1)-\rho\bL\otimes\mathbf{I}_p\big]\be(t-1)\nonumber \\
+ &\bPhi^{-\frac{1}{2}}(t)\bH(t)\bC(t)\beps(t).
\end{align}
From \eqref{eq:matrixform-2}, we have ($\otimes$ denotes Kronecker product)
\begin{align*}
&E[\be^T(t)\bPhi(t)\be(t)] \nonumber\\
 =&  E[\be^T(t-1)(\bPhi(t-1)-\rho\bL\otimes\mathbf{I}_p)^T\bPhi^{-1}(t)\nonumber\\
& \hspace{1em} \times (\bPhi(t-1)-\rho\bL\otimes\mathbf{I}_p)\be(t-1)]\nonumber\\
+& 2E[\be^T(t-1)(\bPhi(t-1)-\rho\bL\otimes\mathbf{I}_p)^T\bPhi^{-1}(t) \bH(t)\bC(t)\beps(t)]\nonumber\\
+& E[\beps^T(t)\bC^T(t)\bH^T(t)\bPhi^{-1}(t)\bH(t)\bC(t)\beps(t)].
\end{align*}
Since $\bC(t)$ and $\beps(t)$ are irrelevant under (as1), the
second term on the right hand side is zero; hence,
\begin{align}\label{eq:conv3}
 &E[\be^T(t)\bPhi(t)\be(t)] \nonumber\\
= & E[\be^T(t-1)(\bPhi(t-1)-\rho\bL\otimes\mathbf{I}_p)^T\bPhi^{-1}(t)\nonumber\\
 \times & (\bPhi(t-1)-\rho\bL\otimes\mathbf{I}_p)\be(t-1)]\nonumber\\
 + &
 E[\beps^T(t)\bC^T(t)\bH^T(t)\bPhi^{-1}(t)\bH(t)\bC(t)\beps(t)].
\end{align}

\textbf{Stage 3.} Consider the first term on the right hand side of
\eqref{eq:conv3}. Since $\bL$ is positive
semi-definite, we can find a matrix
$\mathbf{U}=(\bL\otimes\mathbf{I}_p)^{\frac{1}{2}}$ such that
$\bL\otimes\mathbf{I}_p = \mathbf{U}^T\mathbf{U}$. By the matrix
inversion lemma, it holds that
\begin{align}\label{eq:ineq_inv_phi}
&(\bPhi(t-1)-\rho\bL\otimes\mathbf{I}_p)^{-1}\nonumber\\
= & (\bPhi(t-1)-\rho\mathbf{U}^T\mathbf{U})^{-1}\nonumber\\
= & \bPhi^{-1}(t-1)+\rho\bPhi^{-1}(t-1)\mathbf{U}^T \nonumber\\
\times &
(\mathbf{I}_{Jp}-\rho\mathbf{U}\bPhi^{-1}(t-1)\mathbf{U}^T)^{-1}\mathbf{U}\bPhi^{-1}(t-1).
\end{align}
For $\lambda = 1$, it follows from \eqref{eq:d-rls-phi} that
$\bPhi^{-1}(t-1)- \bPhi^{-1}(t) \succeq \mathbf{0}_{Jp}$. Since
$\bPhi^{-1}(0) = \gamma \mathbf{I}_{Jp}$, it holds that
$\bPhi^{-1}(t-1) \preceq \gamma \mathbf{I}_{Jp}$ for all $t \geq
1$, and consequently
\begin{align*}
\hspace{-0.5em}
\mathbf{I}_{Jp}-\rho\mathbf{U}\bPhi^{-1}(t-1)\mathbf{U}^T
         \succeq
         \mathbf{I}_{Jp}-\rho\gamma\mathbf{U}\mathbf{U}^T =
         \mathbf{I}_{Jp}-\rho\gamma\bL\otimes\mathbf{I}_p.
\end{align*}
If $0<\rho<1/(\gamma\lambda_{\max}(\bL))$, then for all $t
\geq 1$ it follows that
\begin{align*}
         & \mathbf{I}_{Jp}-\rho\mathbf{U}\bPhi^{-1}(t-1)\mathbf{U}^T
         \succeq \mathbf{0}_{Jp}.
\end{align*}
This implies that the second term of \eqref{eq:ineq_inv_phi} is
positive definite. Thus, we have
\begin{align}\label{eq:one-term-new-001}
\bPhi^{-1}(t) \preceq \bPhi^{-1}(t-1) \preceq
(\bPhi(t-1)-\rho\bL\otimes\mathbf{I}_p)^{-1}
\end{align}
and hence, the first term on the right hand side of
\eqref{eq:conv3} is bounded by
\begin{align}\label{eq:one-term-new}
& E[\be^T(t-1)(\bPhi(t-1)-\rho\bL\otimes\mathbf{I}_p)^T\bPhi^{-1}(t)\nonumber\\
& \hspace{1em} \times (\bPhi(t-1)-\rho\bL\otimes\mathbf{I}_p)\be(t-1)]\nonumber\\
\leq & E[\be^T(t-1)(\bPhi(t-1)-\rho\bL\otimes\mathbf{I}_p)^T \be(t-1)]\nonumber\\
\leq & E[\be^T(t-1) \bPhi(t-1)^T \be(t-1)].
\end{align}

\textbf{Stage 4.} Now consider the second term on the right hand
side of \eqref{eq:conv3}. Manipulating the expectation yields
\begin{align*}
&E[\beps^T(t)\bC^T(t)\bH^T(t)\bPhi^{-1}(t)\bH(t)\bC(t)\beps(t)]\\
=&E[\mathrm{tr}(\beps^T(t)\bC^T(t)\bH^T(t)\bPhi^{-1}(t)\bH(t)\bC(t)\beps(t))]\\
=&E[\mathrm{tr}(\bC^T(t)\bH^T(t)\bPhi^{-1}(t)\bH(t)\bC(t)\beps(t)\beps^T(t))]\\
=&E[\mathrm{tr}(\bC^T(t)\bH^T(t)\bPhi^{-1}(t)\bH(t)\bC(t)\mathrm{diag}(\{\sigma_j^2\})].
\end{align*}
where $\mathrm{diag}(\{\sigma_j^2\}) \in \mathbb{R}^{J \times J}$
is a diagonal matrix constructed with $\{\sigma_j^2\}_{j=1}^J$ on
its diagonal. Expanding the matrix multiplications and noting that
$c_j(t) \leq 1$, we obtain
\begin{align}
     & E[\beps^T(t)\bC^T(t)\bH^T(t)\bPhi^{-1}(t)\bH(t)\bC(t)\beps(t)] \nonumber \\
\leq & \sum_{j=1}^{J}\sigma_j^2 E[\h_j^T(t)\bPhi_j^{-1}(t)\h_j(t)]
\nonumber.
\end{align}
Because $\bPhi_j^{-1}(t - 1) \succeq \bPhi_j^{-1}(t)$ due to
\eqref{eq:phi-decreasing}, we further have
\begin{align} \label{eq:two-term-new-001}
     & E[\beps^T(t)\bC^T(t)\bH^T(t)\bPhi^{-1}(t)\bH(t)\bC(t)\beps(t)] \nonumber \\
\leq & \sum_{j=1}^{J}\sigma_j^2 E[\h_j^T(t)\bPhi_j^{-1}(t-1)\h_j(t)] \nonumber \\
\leq & \sum_{j=1}^{J}\sigma_j^2
E[\lambda_{\max}(\bPhi_j^{-1}(t-1)) \|\h_j(t)\|^2].
\end{align}
Since $\bPhi_j^{-1}(t-1)$ and $\h_j(t)$ are independent, it holds
$\forall t > t_0$ that
\begin{align} \label{eq:two-term-new-002}
     & E[\lambda_{\max}(\bPhi_j^{-1}(t-1)) \|\h_j(t)\|^2] \nonumber \\
   = & E[\lambda_{\max}(\bPhi_j^{-1}(t-1))] E[\|\h_j(t)\|^2] \nonumber \\
   < & \frac{\lambda_{\max}(\Rhj^{-1})}{2Q(\tau)(t-1)}
\mathrm{tr}(\Rhj).
\end{align}
The inequality is due to \eqref{eq:ineq_phi-002} that shows
$E[\lambda_{\max}( \bPhi^{-1}_j(t)) ] < \lambda_{\max}(\Rhj^{-1})
/ (2Q(\tau) t)$, $\forall t \geq t_0$ and the fact
$E[\|\h_j(t)\|^2] = \mathrm{tr}(E[\h_j(t)\h_j^T(t)]) =
\mathrm{tr}(\Rhj)$. Using \eqref{eq:two-term-new-002} allows one
to deduce from \eqref{eq:two-term-new-001} that
\begin{align} \label{eq:two-term-new}
     & E[\beps^T(t)\bC^T(t)\bH^T(t)\bPhi^{-1}(t)\bH(t)\bC(t)\beps(t)] \nonumber \\
\leq &
\frac{1}{2Q(\tau)(t-1)}\sum_{j=1}^{J}\sigma_j^2\lambda_{\max}(\Rhj^{-1})\mathrm{tr}(\Rhj)
\end{align}
holds $\forall t > t_0$.

For $t \leq t_0$, we have $\bPhi_j^{-1}(t) \preceq \bPhi_j^{-1}(0)
= \gamma \mathbf{I}_p$ because to \eqref{eq:one-term-new-001}, and
thus
\begin{align*}
 & \sum_{j=1}^{J}\sigma_j^2 E[\h_j^T(t)\bPhi_j^{-1}(t)\h_j(t)] \\
\leq & \sum_{j=1}^{J}\gamma \sigma_j^2 E[\h_j^T(t)\h_j(t)]=
\gamma\sum_{j=1}^{J}\sigma_j^2 tr(\Rhj).
\end{align*}
Therefore, for $t \leq t_0$ \eqref{eq:two-term-new-001} yields
\begin{align}\label{eq:two-term-new2}
     & E[\beps^T(t)\bC^T(t)\bH^T(t)\bPhi^{-1}(t)\bH(t)\bC(t)\beps(t)] \nonumber \\
\leq & \gamma\sum_{j=1}^{J}\sigma_j^2 \mathrm{tr}(\Rhj).
\end{align}

\textbf{Stage 5.} Substituting \eqref{eq:one-term-new},
\eqref{eq:two-term-new}  and \eqref{eq:two-term-new2} into
\eqref{eq:conv3} implies for $t> t_0$ that
\begin{align}\label{eq:ineq_main}
&E[\be^T(t)\bPhi(t)\be(t)]\nonumber\\
\leq &
E[\be^T(t-1)\bPhi(t-1)\be(t-1)]\nonumber\\
+ &
\frac{1}{2Q(\tau)(t-1)}\sum_{j=1}^{J}\sigma_j^2\lambda_{\max}(\Rhj^{-1})\mathrm{tr}(\Rhj)
\end{align}
while for $t \leq t_0$
\begin{align}\label{eq:ineq_main0}
&E[\be^T(t)\bPhi(t)\be(t)]\nonumber\\
\leq & E[\be^T(t-1)\bPhi(t-1)\be(t-1)] + \gamma
\sum_{j=1}^{J}\sigma_j^2 \mathrm{tr}(\Rhj).
\end{align}
Summing \eqref{eq:ineq_main} from $r=t_0+1$ to $r=t$ and
\eqref{eq:ineq_main0} from $r=1$ to $r=t_0$, applying telescopic
cancellation, and noticing that $\bPhi(0) = \gamma^{-1}
\mathbf{I}_{Jp}$, yields for  $t > t_0$
\begin{align} \label{eq:aaaa-0001}
&E[\be^T(t)\bPhi(t)\be(t)] \\
\leq & \gamma^{-1}||\be(0)||^2
+ (\gamma t_0 +\sum_{r=t_0+1}^{t}\frac{\lambda_{\max}(\Rhj^{-1})}{2Q(\tau)(t-1)})\sum_{j=1}^{J}\sigma_j^2\mathrm{tr}(\Rhj) \nonumber \\
\leq & \gamma^{-1}||\be(0)||^2+(\gamma
t_0+\frac{\lambda_{\max}(\Rhj^{-1})}{2Q(\tau)}\ln(t))\sum_{j=1}^{J}\sigma_j^2\mathrm{tr}(\Rhj).
\nonumber
\end{align}

On the other hand, it holds
\begin{align*}
E[\be^T(t)\bPhi(t)\be(t)] & \geq E[\|\be(t)\|^2/\lambda_{\max}(\bPhi^{-1}(t))] \\
& \geq E[\|\be(t)\|]^2/E[\lambda_{\max}(\bPhi^{-1}(t))]
\end{align*}
where the last line is due to Cauchy-Schwarz inequality
\begin{align*}
     & E[\|\be(t)\|^2/\lambda_{\max}(\bPhi^{-1}(t))] E[\lambda_{\max}(\bPhi^{-1}(t))] \\
=    & E[\left(\|\be(t)\|/\lambda_{\max}(\bPhi^{-1}(t))^{\frac{1}{2}}\right)^2] E[\left(\lambda_{\max}(\bPhi^{-1}(t))^{\frac{1}{2}} \right)^2] \\
\geq & E[\|\be(t)\|]^2.
\end{align*}
From \eqref{eq:ineq_phi-002}, $E[\lambda_{\max}( \bPhi^{-1}_j(t))
] < \lambda_{\max}(\Rhj^{-1}) / (2Q(\tau) t) =
1/(\lambda_{\min}(\Rhj) 2Q(\tau) t) $ holds asymptotically.
Definining $\mu := \min\{\lambda_{\min}(\Rhj), j \in
\mathcal{V}\}$, we establish that
\begin{align} \label{eq:aaaa-0002}
& 2Q(\tau) \mu t E[||\be(t)||^2] \leq  E[\be^T(t)\bPhi(t)\be(t)],
\quad t > t_0.
\end{align}

Combining \eqref{eq:aaaa-0001} and \eqref{eq:aaaa-0002} implies
\begin{align} \label{eq:ineq_main1}
& 2Q(\tau) \mu t E[||\be(t)||]^2 \\
\leq & \gamma^{-1}||\be(0)||^2+(\gamma t_0 +
\frac{\lambda_{\max}(\Rhj^{-1})}{2Q(\tau)}\ln(t))\sum_{j=1}^{J}\sigma_j^2\mathrm{tr}(\Rhj).
\nonumber
\end{align}
Finally, with $||\be(t)||^2:= \sum_{j=1}^J ||\be_j(t)||^2
 = \sum_{j=1}^J ||\s_j(t)-\s_0||^2$ this
leads to \eqref{eq:main}, which completes the proof of CD-RLS-1.

Consider next CD-RLS-2. Stage 1 of the proof remains the same, while
for Stage 2,
$\be_j(t-1)-\be_{j'}(t-1)$ is replaced by $c_j(t) \big[
\be_j(t-1)-\be_{j'}(t-1) \big]$ in \eqref{eq:conv2} to arrive at
\begin{align}\label{eq:conv2-2}
& \bPhi_j(t)\be_j(t) \\
 = & \bPhi_j(t-1)\be_j(t-1) -c_j(t) \h_j(t)\h_j^T(t)\be_j(t-1)\nonumber\\
+ & c_j(t)\h_j(t)\epsilon_j(t) -\rho\sum_{j'\in\N_j} c_j(t) \big[
\be_j(t-1)-\be_{j'}(t-1) \big]. \nonumber
\end{align}
Its matrix form \eqref{eq:conv3} can be expressed as
\begin{align}\label{eq:conv3-2}
 &E[\be^T(t)\bPhi(t)\be(t)] \nonumber\\
=& E[\be^T(t-1)(\bPhi(t-1)-\rho(\bC(t)\bL)\otimes\mathbf{I}_p)^T\bPhi^{-1}(t)\nonumber\\
& \hspace{1em} \times (\bPhi(t-1)-\rho(\bC(t)\bL)\otimes\mathbf{I}_p)\be(t-1)]\nonumber\\
+&
 E[\beps^T(t)\bC^T(t)\bH^T(t)\bPhi^{-1}(t)\bH(t)\bC(t)\beps(t)].
\end{align}

Observe that the right hand sides of \eqref{eq:conv3} and
\eqref{eq:conv3-2} are only different in their first terms.
Similar to Stage 3
(cf. \eqref{eq:one-term-new}), we need to show that the first term
satisfies
\begin{align}\label{eq:goal2}
& E[\be^T(t-1)(\bPhi(t-1)-\rho(\bC(t)\bL)\otimes\mathbf{I}_p)^T\bPhi^{-1}(t)\nonumber\\
& \hspace{1em} \times (\bPhi(t-1)-\rho(\bC(t)\bL)\otimes\mathbf{I}_p)\be(t-1)]\nonumber\\
\leq & E[\be^T(t-1)\bPhi(t-1)\be(t-1)].
\end{align}
Substituting the update
\eqref{eq:phi-decreasing} with $\lambda = 1$ into
\eqref{eq:goal2}, it suffices to prove that
\begin{align}\label{eq:ineq2}
& E[\be^T(t-1)\bC(t)\otimes \mathbf{I}_p\bH(t)\bH^T(t) \\
& \hspace{1em} \times (\mathbf{I}_J+\bH^T(t)\bPhi^{-1}(t-1)\bH(t))^{-1}\otimes\mathbf{I}_p\be(t-1)]\nonumber\\
\geq & \rho E[\be^T(t-1) \mathbf{W} \be(t-1)] \nonumber
\end{align}
where
\begin{align*}
\mathbf{W}   := & \mathbf{W}_1+\mathbf{W}_1^T-\mathbf{W}_2-(\bL\bC(t))\otimes\mathbf{I}_p-(\bC(t)\bL)\otimes\mathbf{I}_p\nonumber\\
+ & \rho\bL\otimes\mathbf{I}_p\bPhi^{-1}(t-1)(\bC(t)\bL)\otimes\mathbf{I}_p \\
\mathbf{W}_1 := & \bC(t)\otimes\mathbf{I}_p\bH(t)\bH^T(t)\bPhi^{-1}(t-1)\\
\times & ((\mathbf{I}_J+\bH^T(t)\bPhi^{-1}(t-1)\bH(t))^{-1}\bL)\otimes\mathbf{I}_p\\
\mathbf{W}_2 := & (\bL\bC(t))\otimes\mathbf{I}_p\bPhi^{-1}(t-1)\bH(t)\bH^T(t)\bPhi^{-1}(t-1)\\
\times &
((\mathbf{I}_J+\bH^T(t)\bPhi^{-1}(t-1)\bH(t))^{-1}\bL)\otimes\mathbf{I}_p.&
\end{align*}

For the left hand side of \eqref{eq:ineq2}, use the lower bound of
the conditional expectation $2Q(\tau)\leq
E[c_j(t)|\h_j(t),\s_j(t-1)]$ to eliminate $\bC(t)$, and arrive at
\begin{align} \label{eq:one-term-alg-2-001}
& E[\be^T(t-1)\bC(t)\otimes \mathbf{I}_p\bH(t)\bH^T(t)\\
& \hspace{1em} \times (\mathbf{I}_J+\bH^T(t)\bPhi^{-1}(t-1)\bH(t))^{-1}\otimes\mathbf{I}_p\be(t-1)] \nonumber \\
\geq & 2Q(\tau) E[\be^T(t-1) \bH(t)\bH^T(t) \nonumber \\
& \hspace{1em} \times
(\mathbf{I}_J+\bH^T(t)\bPhi^{-1}(t-1)\bH(t))^{-1}\otimes\mathbf{I}_p\be(t-1)].
\nonumber
\end{align}
By \eqref{eq:one-term-new-001}, it holds that $\bPhi^{-1}(t-1) \preceq
\bPhi^{-1}(0) = \gamma \mathbf{I}_{Jp}$, and thus
\begin{align*}
& \left[\mathbf{I}_J+\bH^T(t)\bPhi^{-1}(t-1)\bH(t)\right]^{-1} \succeq
\left[\mathbf{I}_J+ \gamma  \bH^T(t) \bH(t)\right]^{-1}.
\end{align*}
By assumption $\{ {\bf h}_j(t)\}$ are uniformly bounded. If
${\bf h}_j^T(t) {\bf h}_j(t) \leq K$ for all $j = 1, \ldots, J$, we find
\begin{align} \label{eq:one-term-alg-2-002}
& \left[\mathbf{I}_J+\bH^T(t)\bPhi^{-1}(t-1)\bH(t)\right]^{-1} \succeq
\frac{1}{1+\gamma K^2}  \mathbf{I}_J.
\end{align}
Substituting \eqref{eq:one-term-alg-2-002} into
\eqref{eq:one-term-alg-2-001}, we obtain a lower bound for the
left hand side of \eqref{eq:ineq2} given by
\begin{align} \label{eq:one-term-alg-2-002-1}
& E[\be^T(t-1)\bC(t)\otimes \mathbf{I}_p\bH(t)\bH^T(t)\\
& \hspace{1em} \times (\mathbf{I}_J+\bH^T(t)\bPhi^{-1}(t-1)\bH(t))^{-1}\otimes\mathbf{I}_p\be(t-1)] \nonumber \\
\geq & \frac{2Q(\tau)}{1+\gamma K^2}E[\be^T(t-1)\bH(t)\bH^T(t)\be(t-1)] \nonumber \\
= & \frac{2Q(\tau)}{1+\gamma K^2}E[\be^T(t-1)\text{diag}\{\Rhj\}\be(t-1)] \nonumber \\
\geq & \frac{2Q(\tau)\mu}{1+\gamma K^2}E[||\be(t-1)||^2].
\nonumber
\end{align}

As for the right hand side of \eqref{eq:ineq2}, it is upper bounded
by
\begin{align} \label{eq:one-term-alg-2-003}
& \rho E[\be^T(t-1) \mathbf{W} \be(t-1)] \\
\leq  & \rho E[(2||\mathbf{W}_1||_2+||\mathbf{W}_2||_2+2||\bL||_2 \nonumber \\
& \hspace{1em} +
\rho||\bL||_2^2||\bPhi^{-1}(t-1)||_2)||\be(t-1)||^2] \nonumber
\end{align}
where we used that all the diagonal elements $c_j(t)$ of $\bC(t)$ are within the range $[0, 1]$ while $||\mathbf{W}_1||_2$ is upper bounded by
\begin{align*}
||\mathbf{W}_1||_2 \leq &  ||\bC(t)||_2  ||\bH(t)||_2^2 ||\bPhi^{-1}(t-1)||_2 \\
\times & ||(\mathbf{I}_J+\bH^T(t)\bPhi^{-1}(t-1)\bH(t))^{-1}||_2
||\bL||_2.
\end{align*}
Noticing that $||\bC(t)||_2 \leq 1$, $||\bH(t)||_2^2 \leq K^2$ by
assumption, $||\bPhi^{-1}(t-1)||_2 \leq ||\bPhi^{-1}(0)||_2 =
\gamma$, $||(\mathbf{I}_J+\bH^T(t)\bPhi^{-1}(t-1)\bH(t))^{-1}||_2
\leq 1$ and $||\bL||_2 \leq \lambda_{\max}(\bL)$, we find that
\begin{align*}
& ||\mathbf{W}_1||_2 \leq \gamma\lambda_{\max}(\bL)K^2.
\end{align*}
Similarly, $||\mathbf{W}_2||_2$ is upper bounded by
\begin{align*}
& ||\mathbf{W}_2||_2 \leq \gamma^2\lambda_{\max}(\bL)^2K^2.
\end{align*}
Therefore, \eqref{eq:one-term-alg-2-003} reduces to
\begin{align} \label{eq:one-term-alg-2-004}
& \rho E[\be^T(t-1) \mathbf{W} \be(t-1)] \\
\leq & \rho(2\gamma\lambda_{\max}(\bL)K^2+\gamma^2\lambda_{\max}(\bL)^2K^2+2\lambda_{\max}(\bL) \nonumber \\
+ & \rho \gamma \lambda_{\max}(\bL)^2)E[||\be(t-1)||^2]. \nonumber
\end{align}

Considering a positive constant
\begin{align*}
\rho_0 := & \sqrt{\frac{2Q(\tau)\mu}{\gamma\lambda_{\max}(\bL)^2(1+\gamma K^2)}+ (\frac{\gamma K^2}{2}+\frac{\gamma K^2 + 1}{\gamma \lambda_{\max}(\bL)})^2}\\
- & (\frac{\gamma K^2}{2}+\frac{\gamma K^2 + 1}{\gamma
\lambda_{\max}(\bL)})
\end{align*}
and combining \eqref{eq:one-term-alg-2-002-1} with
\eqref{eq:one-term-alg-2-004}, we see that if $\rho$ is chosen
within $[0, \rho_0]$, then \eqref{eq:ineq2} holds for all $t \geq
1$; and so does \eqref{eq:goal2}.

Following Stages 4 and 5 in the proof for CD-RLS-1, we can show
that \eqref{eq:main} holds almost surely for CD-RLS-2 $\forall t
> t_0$. This completes the proof of the entire theorem.
\end{proof}

\section{Proof of Theorem \ref{thm:convergence3}}
\label{sec:proof-2}
%
%
Theorem \ref{thm:convergence3} relies on the following lemma.
\begin{lemma}\label{lem:cd-rls-3}
There exist constants $M>0$ and $t_0 > 0$ such that
\begin{align}\label{eq:lem}
\hspace{-1em} E[||\be_j(t)-\be_j(t-1)||] \leq \frac{M}{t}, ~
\forall j = 1,\cdots,J, \ t\geq t_0.
\end{align}
\end{lemma}
\begin{proof}[Proof of Lemma \ref{lem:cd-rls-3}]
The update of $\be_j(t)$ for CD-RLS-3 is (cf. \eqref{eq:conv1} for
CD-RLS-1)
\begin{align} \label{eq:lemma-1-001}
\be_j(t) &=\be_j(t-1)-c_j(t)\bPhi_j^{-1}(t)\h_j(t)\h_j^T(t)\be_j(t-1)\\
& +c_j(t)\bPhi_j^{-1}(t)\h_j(t)\epsilon_j(t) \nonumber \\
&
-c_j(t)\rho\bPhi_j^{-1}(t)\sum_{j'\in\N_j}(\be_j(t-1)-\be_{j'}(t-d_j^{j'}(t))).
\nonumber
\end{align}
Per time $t$, $t-d_j^{j'}(t)$ is the latest time slot when node
$j$ received information from its neighbor $j'$. Therefore,
$d_j^{j'}(t)$ can be viewed as network delay caused by the
censoring strategy. Then we have
\begin{align*}
&||\be_j(t)-\be_j(t-1)|| \\
= &||c_j(t)\bPhi_j^{-1}(t)\big[\h_j(t)\h_j^T(t)\be_j(t-1) -\h_j(t)\epsilon_j(t)\\
& \hspace{0.5em} + \rho\sum_{j'\in\N_j}(\be_j(t-1)-\be_{j'}(t-d_j^{j'}(t)))\big]||\\
\leq & ||\bPhi_j^{-1}(t)||_2 \big[||\h_j(t)||^2||\be_j(t-1)||+||\h_j(t)|||\epsilon_j(t)|\\
+ &
\rho\sum_{j'\in\N_j}(||\be_j(t-1)||+||\be_{j'}(t-d_j^{j'}(t))||)\big].
\end{align*}
In deriving the inequality we use the fact that $c_j(t) \in \{0,
1\}$.

According to \eqref{eq:ineq_phi-002} in the proof of Theorem
\ref{thm:convergence}, which also holds true for CD-RLS-3, there
exists $t_0 > 0$, such that $E[||\bPhi_j^{-1}(t)||_2]$ is upper bounded by $M_1/t$ when $t > t_0$, where $M_1$ is
a positive constant determined by $Q(\tau)$ and the smallest
eigenvalue of $\Rhj(t)$. By (as1) and (as2), $||\h_j(t)||$,
$||\be_j(t-1)||$ and $||\be_{j'}(t-d_j^{j'}(t))||$ are also upper
bounded. Therefore, there exist constants $M_2, M_3 > 0$,
such that
\begin{align*}
&||\be_j(t)-\be_j(t-1)||\leq||\bPhi_j^{-1}(t)||_2[M_2
+M_3|\epsilon_j(t)|].
\end{align*}
Taking expectations on both sides yields \eqref{eq:lem}.
\end{proof}

Now we turn to prove Theorem \ref{thm:convergence3}.

\begin{proof}[Proof of Theorem \ref{thm:convergence3}] Rewrite the update of $\be_j(t)$ for CD-RLS-3 in
\eqref{eq:lemma-1-001} to
\begin{align*}
\be_j(t) &=\be_j(t-1)-c_j(t)\bPhi_j^{-1}(t)\h_j(t)\h_j^T(t)\be_j(t-1)\\
&-c_j(t)\rho\bPhi_j^{-1}(t)\sum_{j'\in\N_j}(\be_j(t-1)-\be_{j'}(t-1))\\
&-c_j(t)\rho\bPhi_j^{-1}(t)\sum_{j'\in\N_j}(\be_{j'}(t-1)-\be_{j'}(t-d_j^{j'}(t)))\\
&+c_j(t)\bPhi_j^{-1}(t)\h_j(t)\epsilon_j(t).
\end{align*}
Multiplying $\bPhi_j(t)$ on both sides, we have
\begin{align*}
\bPhi_j(t)\be_j(t) &= \bPhi_j(t-1)\be_j(t-1)\\
&-c_j(t)\rho\sum_{j'\in\N_j}(\be_j(t-1)-\be_{j'}(t-1))\\
&-c_j(t)\rho\sum_{j'\in\N_j}(\be_{j'}(t-1)-\be_{j'}(t-d_j^{j'}(t)))\\
&+c_j(t)\h_j(t)\epsilon_j(t).
\end{align*}
Using the same notations as in the proof of Theorem
\ref{thm:convergence}, we obtain an matrix form
\begin{align}\label{eq:update3}
\bPhi(t)\be(t) &= (\bPhi(t-1)-\rho(\bC(t)\bL)\otimes\mathbf{I}_p)\be(t-1)\\
 & + \bH(t)\bC(t)\beps(t)
 -\rho(\bC(t)\otimes\mathbf{I}_p)\tilde{\be}(t). \nonumber
\end{align}
where $\tilde{\be}(t) \in \mathbb{R}^{Jp}$ and its $j$th block is
$\sum_{j'\in\N_j}(\be_{j'}(t-1)-\be_{j'}(t-d_j^{j'}(t)))$. Observe
that $\tilde{\be}(t)$ contains the differences between the local
estimates and their delayed values, and hence plays a critical
role in the convergence proof. Below we look for an upper bound
for $E[||\tilde{\be}(t)||]$.

By the Cauchy-Schwarz inequality, we have
\begin{align*}
& ||\sum_{j'\in\N_j}(\be_{j'}(t-1)-\be_{j'}(t-d_j^{j'}(t)))||\\
\leq& \sqrt{|\N_j|}\sqrt{\sum_{j'\in\N_j}||\be_{j'}(t-1)-\be_{j'}(t-d_j^{j'}(t))||^2}\\
= & \sqrt{|\N_j|}\sqrt{\sum_{j'\in\N_j}||\sum_{k=1}^{d_j^{j'}(t)-1}(\be_{j'}(t-k)-\be_{j'}(t-k-1))||^2}\\
\leq &
\sqrt{|\N_j|(d_{\max}-1)}\sqrt{\sum_{j'\in\N_j}\sum_{k=1}^{d_{\max}-1}||\be_{j'}(t-k)-\be_{j'}(t-k-1)||^2}\\
\leq &
\sqrt{|\N_j|(d_{\max}-1)}\sum_{j'\in\N_j}\sum_{k=1}^{d_{\max}-1}||\be_{j'}(t-k)-\be_{j'}(t-k-1)||.
\end{align*}
Here we use the fact that $d_j^{j'}(t)$ is no larger than the
maximal delay $d_{\max}$. Take expectation and use Lemma
\ref{lem:cd-rls-3}. There exists $t_0 > 0$ such that when $t\geq
t_0$ it holds
\begin{align*}
& E[||\sum_{j'\in\N_j}(\be_{j'}(t-1)-\be_{j'}(t-d_j^{j'}(t)))||]\\
\leq & \sqrt{|\N_j|(d_{\max}-1)}\sum_{j'\in\N_j}\sum_{k=1}^{d_{\max}-1}\frac{M}{t-k}\\
\leq & (|\N_j|(d_{\max}-1))^{\frac{3}{2}}\frac{M}{t-d_{\max}}.
\end{align*}
Therefore, $\forall t\geq t_0$
\begin{align*}
E[||\tilde{\be}(t)||] &= E[\sqrt{\sum_{j=1}^{J}||\sum_{j'\in\N_j}(\be_{j'}(t-1)-\be_{j'}(t-d_j^{j'}(t)))||^2}]\\
&\leq E[\sum_{j=1}^{J}||\sum_{j'\in\N_j}(\be_{j'}(t-1)-\be_{j'}(t-d_j^{j'}(t)))||]\\
&\leq (\sum_{j=1}^{J}|\N_j|^{\frac{3}{2}})\frac{(d_{\max}-1)^{\frac{3}{2}}M}{t-d_{\max}}\\
&\leq \frac{M_0}{t}
\end{align*}
for some constant $M_0>0$.

Back to \eqref{eq:update3}, multiplying $\bPhi^{-\frac{1}{2}}(t)$
on both sides yields
\begin{align*}
& \bPhi^{\frac{1}{2}}(t)\be(t) = \bPhi^{-\frac{1}{2}}(t)(\bPhi(t-1)-\rho(\bC(t)\bL)\otimes\mathbf{I}_p)\be(t-1)\nonumber\\
& \hspace{2.4em} +\bPhi^{-\frac{1}{2}}(t)\bH(t)\bC(t)\beps(t)
-\rho\bPhi^{-\frac{1}{2}}(t)(\bC(t)\otimes\mathbf{I}_p)\tilde{\be}(t).
\end{align*}
Since $\bH(t)$ and $\beps(t)$ are independent as given by (as1),
we have
\begin{align}\label{eq:ineq_main3-temp}
&E[\be^T(t)\bPhi(t)\be(t)] \\
=& E[\be^T(t-1)(\bPhi(t-1)-\rho(\bC(t)\bL)\otimes\mathbf{I}_p)^T\bPhi^{-1}(t)\nonumber\\
& \hspace{1em} \times (\bPhi(t-1)-\rho(\bC(t)\bL)\otimes\mathbf{I}_p)\be(t-1)]\nonumber\\
+& E[\beps^T(t)\bC^T(t)\bH^T(t)\bPhi^{-1}(t)\bH(t)\bC(t)\beps(t)] \nonumber \\
+& \rho^2E[\tilde{\be}^T(t)\bC(t)\otimes\mathbf{I}_p\bPhi^{-1}(t)\bC(t)\otimes\mathbf{I}_p\tilde{\be}(t)] \nonumber\\
+& \rho E[\tilde{\be}^T(t)\bC(t)\otimes\mathbf{I}_p\bPhi^{-1}(t)\nonumber\\
& \hspace{1em} \times
(\bPhi(t-1)-\rho(\bC(t)\bL)\otimes\mathbf{I}_p)\be(t-1)].
\nonumber
\end{align}
Observe that \eqref{eq:ineq_main3-temp} is different to
\eqref{eq:conv3-2} for having the last two terms at the right hand
side. Because all the diagonal elements $c_j(t)$ in the diagonal
matrix $\bC(t)$ are within $[0,1]$, $\forall t \geq t_0$
\begin{align*}
     & \rho^2E[\tilde{\be}^T(t)\bC(t)\otimes\mathbf{I}_p\bPhi^{-1}(t)\bC(t)\otimes\mathbf{I}_p\tilde{\be}(t)] \\
\leq & \rho^2 E[||\tilde{\be} (t)||^2 ||\bPhi^{-1}(t)||_2] \\
\leq & \frac{\rho^2  M_0^2}{t^2} E[||\bPhi^{-1}(t)||_2].
\end{align*}
The right hand side is in the order of $O(1/t^3)$ because
$E[||\bPhi^{-1}(t)||_2]$ is no larger than
$\lambda_{\max}(\Rhj^{-1})/(2Q(\tau)t)$ for all $t \geq t_0$ as we
have shown in Step 1 of the proof of Theorem \ref{thm:convergence}
(cf. \eqref{eq:ineq_phi-002}). Meanwhile, $\forall t \geq t_0$
\begin{align*}
& \rho E[\tilde{\be}^T(t)\bC(t)\otimes\mathbf{I}_p\bPhi^{-1}(t)\nonumber\\
& \hspace{1em} \times
(\bPhi(t-1)-\rho(\bC(t)\bL)\otimes\mathbf{I}_p)\be(t-1)] \nonumber
\\ \leq & \rho E[ \|\tilde{\be}(t)\| ||\bPhi^{-1}(t)||_2
(||\bPhi(t-1)||_2+\rho||\bL||_2) ||\be(t-1)||].
\end{align*}
Observe that $E[\|\tilde{\be}(t)\|]$ and $E[||\bPhi^{-1}(t)||_2]$
are in the orders of $O(1/t)$ and $O(1/t)$, respectively, while
$E[||\bPhi(t-1)||_2+\rho||\bL||_2]$ is in the order of $O(t)$
because $E[||\bPhi_j(t)||_2]\leq t\lambda_{\max}(\Rhj)$ (cf. \eqref{eq:ineq_phi-001}). In
addition, $||\be(t-1)||$ is bounded by (as2). Therefore, the right
hand side is in the order of $O(1/t)$.

For the first term at the right hand side of
\eqref{eq:ineq_main3-temp}, similar to the proof for CD-RLS-2, if
$\rho$ is chosen within $[0,\rho_0]$ we are able to show that (cf.
\eqref{eq:goal2})
\begin{align*}
& E[\be^T(t-1)(\bPhi(t-1)-\rho(\bC(t)\bL)\otimes\mathbf{I}_p)^T\bPhi^{-1}(t)\nonumber\\
 &\hspace{1em} \times (\bPhi(t-1)-\rho(\bC(t)\bL)\otimes\mathbf{I}_p)\be(t-1)]\nonumber\\
\leq & E[\be^T(t-1)\bPhi(t-1)\be(t-1)].
\end{align*}
Finally, following Step 4 of the proof of Theorem
\ref{thm:convergence} to handle the second term at the right hand
side of \eqref{eq:ineq_main3-temp}, we know that it is also in the
order of $O(1/t)$. Therefore, for all $t \geq t_0$
\eqref{eq:ineq_main3-temp} yields
\begin{align*}
&E[\be^T(t)\bPhi(t)\be(t)]\\
\leq & E[\be^T(t-1)\bPhi(t-1)\be(t-1)]
+\frac{K_1}{t}+\frac{K_2}{t^3}.
\end{align*}
where $K_1, K_2 > 0$ are constants. Summing up both sides from
time $r = t_0$ to $r = t$, we have
\begin{align} \label{eq:ineq_main3-temp-1}
&E[\be^T(t)\bPhi(t)\be(t)]\\
\leq & E[\be^T(t_0-1)\bPhi(t_0-1)\be(t_0-1)] +\sum_{r=t_0}^t
\frac{K_1}{t}+\sum_{r=t_0}^t\frac{K_2}{t^3}. \nonumber
\end{align}
Observing that $E[\be^T(t_0-1)\bPhi(t_0-1)\be(t_0-1)]$ is bounded
because $||\be(t_0-1)||$ is bounded by (as2), the right hand side
of \eqref{eq:ineq_main3-temp-1} is in the order of
$O(1)+O(\ln(t))$. Following the argument in Step 5 of the proof of
Theorem \ref{thm:convergence}, $E[||\bPhi^{-1}(t)||_2]$ is in the order of
$O(1/t)$ when $t \geq t_0$. Therefore, $E[||\be(t)||]^2$ is in the
order of $O(1/t)+O(\ln(t)/t)$, which completes the proof of
Theorem \ref{thm:convergence3}.
\end{proof}



\begin{thebibliography}{99}
\bibitem{Amari1998}
S. I. Amari. ``Natural gradient works efficiently in learning,''
\emph{Neural Computation}, vol. 10, pp. 251--276, 1998.

\bibitem{Appadwedula2008-Censoring}
S. Appadwedula, V. V. Veeravalli, and D. L. Jones, ``Decentralized
detection with censoring sensors,'' \emph{IEEE Transactions on Signal
    Processing}, vol. 56, pp. 1362--1373, April 2008.

\bibitem{Arroyo2013-censoring}
R. Arroyo-Valles, S. Maleki, and G. Leus, ``A censoring strategy
for decentralized estimation in energy-constrained adaptive
diffusion networks,'' \emph{Proc. of Intl. Work. on Signal
    Processing Advances in Wireless Communications}, Germany, June 2013.

\bibitem{Berberidis2016-Censoring}
D. Berberidis, V. Kekatos, and G. B. Giannakis, ``Online censoring
for large-scale regressions with application to streaming big
data,'' \emph{IEEE Transactions on Signal Processing}, vol. 64, pp.
3854--3867, Aug. 2016.

\bibitem{Bianchi2013}
P. Bianchi, G. Fort, and W. Hachem, ``Performance of a distributed
stochastic approximation algorithm,'' \emph{IEEE Transactions on
Information Theory}, vol. 59, no. 11, pp. 7405-7418, Nov. 2013.

\bibitem{Morral2017}
G. Morral, P. Bianchi and G. Fort, ``Success and Failure of Adaptation-Diffusion Algorithms With Decaying Step Size in Multiagent Networks,'' \emph{IEEE Transactions on Signal Processing}, vol. 65, no. 11, pp. 2798-2813, June 1, 2017.

\bibitem{Cevher2014}
V. Cevher, S. Becker, and M. Schmidt, ``Convex optimization for
big data: Scalable, randomized, and parallel algorithms for big
data analytics,'' \emph{IEEE Signal Processing Magazine}, vol. 31, pp.
32--43, Sept. 2014.


\bibitem{Grimmett2011}
G. Grimmett and D. Stirzaker, \textit{Probability and Random
Processes}, Oxford University Press, 2011.

\bibitem{Giannakis2016}
G. B. Giannakis, Q. Ling, G. Mateos, I. D. Schizas, and H. Zhu,
``Decentralized learning for wireless communications and
networking,'' \textit{Splitting Methods in Communication and
    Imaging, Science and Engineering}, R. Glowinski, S. Osher, and W.
Yin (eds.), Springer, 2016.

\bibitem{Jiang2005-censoring}
R. Jiang, Y. Lin, B. Chen, and B. Suter, ``Distributed sensor
censoring for detection in sensor networks under communication
constraints,'' \emph{Proc. of Asilomar Conf. on Signals, Systems and
    Computers}, Pacific Grove, CA, Nov. 2005.

\bibitem{Jiang2005-onoff}
R. Jiang and B. Chen, ``Fusion of censored decisions in wireless
sensor networks,'' \emph{IEEE Transactions on Wireless Communications},
vol. 4, pp. 2668--2673, Dec. 2005.

\bibitem{Kushner1997-RLS}
H. J. Kushner and G. G. Yin, \textit{Stochastic Approximation
    Algorithms and Applications}, Springer, 1997.

\bibitem{Realdata}
M. Lichman, ``UCI machine learning repository,'' 2013. Available
at: \url{http://archive.ics.uci.edu/ml}

\bibitem{Liu2015-censored-regression}
Z. Liu, C. Li, and Y. Liu, ``Distributed censored regression over
networks,'' \emph{IEEE Transactions on Signal Processing}, vol. 63, pp.
5437--5449, Oct. 2015.

\bibitem{Global-performance}
C. G. Lopes and A. H. Sayed, ``Diffusion least-mean squares over
adaptive networks: Formulation and performance analysis,'' \emph{IEEE
    Transactions on Signal Processing}, vol. 56, pp. 3122--3136, July
2008.

\bibitem{Mateos2012-DRLS}
G. Mateos and G. B. Giannakis, ``Distributed recursive
least-squares: Stability and performance analysis,'' \emph{IEEE
    Transactions on Signal Processing}, vol. 60, pp. 3740--3754, July
2012.

\bibitem{Mateos2009-JASP}
G. Mateos, I. D. Schizas and G. B. Giannakis, ``Performance
Analysis of the Consensus-Based Distributed LMS Algorithm,''
\emph{EURASIP Journal on Advances in Signal Processing}, Article
ID 981030, 2009.

\bibitem{Mateos2009-DRLS}
G. Mateos, I. D. Schizas and G. B. Giannakis, ``Distributed
Recursive Least-Squares for Consensus-Based In-Network Adaptive
Estimation,'' \emph{IEEE Transactions on Signal Processing}, vol. 57, pp.
4583--4588, Nov. 2009.

\bibitem{Msechu2012-cratio}
E. Msechu and G. B. Giannakis, ``Sensor-centric data reduction for
estimation with WSNs via censoring and quantization,'' \emph{IEEE
    Transactions on Signal Processing}, vol. 60, pp. 400--414, Jan.
2012.

\bibitem{Patwari2003-censoring}
N. Patwari, and A. O. Hero, ``Hierarchical censoring for
distributed detection in wireless sensor networks,'' \emph{Proc.
    of Intl. Conf. on Acoustics, Speech, and Signal
    Processing}, Hong Kong, 2003.

\bibitem{Predd2006}
J. Predd, S. Kulkarni, and H. V. Poor, ``Distributed learning in
wireless sensor networks,'' \emph{IEEE Signal Processing Magazine}, vol. 23, pp. 56--69, July 2006.

\bibitem{Rabbat2004-ipsn}
M. Rabbat and R. Nowak, ``Distributed optimization in sensor
networks,'' \emph{Intl. Conf. on Information
Processing in Sensor Networks}, pp. 20-27, Berkeley, CA, April 2004.

\bibitem{Rago1996-censoring}
C. Rago, P. Willett, and Y. Bar-Shalom, ``Censoring sensors: A
low-communication-rate scheme for distributed detection,'' \emph{IEEE
    Transactions on Aerospace and Electronic Systems}, vol. 32, pp.
554--568, Apr. 1996.

\bibitem{Sharkh2013}
M. A. Sharkh, M. Jammal, A. Shami, and A. Ouda, ``Resource
allocation in a network-based cloud computing environment: Design
challenges,'' \emph{IEEE Communications Magazine}, vol. 51, pp. 46--52,
Nov. 2013.

\bibitem{Slavakis2014-RLS}
K. Slavakis, S. J. Kim, G. Mateos, and G. B. Giannakis,
``Stochastic approximation vis-a-vis online learning for big data
analytics,'' \emph{IEEE Signal Processing Magazine}, vol. 31, pp.
124--129, Nov. 2014.

\bibitem{MSD2}
V. Solo and X. Kong, \textit{Adaptive Signal Processing
    Algorithms: Stability and Performance}, \emph{Prentice Hall}, 1995.



\bibitem{Tseng1991}
P. Tseng, ``Applications of a splitting algorithm to decomposition
in convex programming and variational inequalities,'' \emph{SIAM Journal
on Control and Optimization}, vol. 29, pp. 119--138, Jan. 1991.



\bibitem{CD-RLS-ICASSP}
Z. Wang, Z. Yu, Q. Ling, D. Berberidis, and G. B. Giannakis,
``Distributed recursive least-squares with data-adaptive
censoring,'' \emph{ Proc. Conf. on Acoustics, Speech, and Signal
Processing}, New Orleans, March 2017.

\end{thebibliography}
\end{document}